\documentclass[runningheads,11points]{llncs}

\usepackage{amssymb,amstext}
\usepackage{cite,scrtime}
\usepackage[bookmarks=false]{hyperref}
\usepackage{graphics}
\usepackage{amsmath,dsfont}
\usepackage{mathrsfs}
\usepackage{fancyhdr}
\usepackage{verbatim}
\usepackage{boxedminipage}
\usepackage{colortbl}
\usepackage{tikz}

\usepackage[T1]{fontenc}
\usepackage{tikz}
\usepackage{graphicx,xspace}

\usepackage[ruled, vlined,boxed,commentsnumbered]{algorithm2e}

\usepackage{vmargin}
\setmarginsrb{3.7cm}{1.5cm}{3.7cm}{3.0cm}{1.2cm}{1.3cm}{0.9cm}{1.2cm}

\newcommand{\ie}{i.e.,~}

\newcommand{\N}{\mathbb{N}\xspace}
\newcommand{\opt}{{\sf opt}\xspace}
\newcommand{\NP}{{\sf NP}\xspace}

\newcommand{\conf}{{\sf conf}\xspace}

\newcommand{\tdmod}{{\sf tdmod}\xspace}
\newcommand{\twmod}{{\sf twmod}\xspace}
\newcommand{\degmod}{{\sf degmod}\xspace}
\newcommand{\pfm}{{\sf pfm}\xspace}

\newcommand{\CVD}{{\sf CVD}\xspace}

\newcommand{\IS}{\mbox{\sc IS}\xspace}
\newcommand{\VC}{\textsc{VC}\xspace}
\newcommand{\DS}{\textsc{DS}\xspace}

\newcommand{\FVS}{\textsc{FVS}\xspace}
\newcommand{\OCT}{\textsc{OCT}\xspace}
\newcommand{\RBDS}{\textsc{RBDS}\xspace}
\newcommand{\REF}{\textsc{ref}\xspace}
\newcommand{\PPT}{\textsc{ppt}\xspace}

\newcommand{\AND}{\textsc{and}\xspace}
\newcommand{\OR}{\textsc{or}\xspace}

\renewcommand{\sup}{\text{sup}}
\renewcommand{\inf}{\text{inf}}
\newcommand{\dec}{\text{dec}}

\newcommand{\X}{\mathcal{X}}
\renewcommand{\H}{\mathcal{H}}
\newcommand{\A}{A}
\newcommand{\C}{C}
\newcommand{\R}{R}
\newcommand{\Y}{Y}
\newcommand{\x}{x}
\renewcommand{\a}{a}
\renewcommand{\c}{c}
\renewcommand{\r}{r}
\newcommand{\y}{y}

\newcommand{\tw}{{\sf tw}}
\newcommand{\td}{{\sf td}}
\newcommand{\rw}{{\sf rw}}

\renewcommand{\O}{\mathcal{O}}

\newtheorem{observation}{Observation}
\renewenvironment{proof}[1][]{\par \noindent {\bf Proof:#1}\ }{\hfill$\Box$\\}

\title{How much does a treedepth modulator help to obtain polynomial kernels beyond sparse graphs? \thanks{Work supported by the French projects DEMOGRAPH (ANR-16-CE40-0028) and ESIGMA (ANR-17-CE40-0028). An extended abstract of this work appeared in the \emph{Proceedings of the 12th International Symposium on Parameterized and Exact Computation (\textbf{IPEC}), pages 10:1--10:13, volume 89 of LIPIcs, Vienna, Austria, September \textbf{2017}}.}
}

\author{Marin Bougeret~\inst{1}  \and Ignasi Sau~\inst{2}}

\authorrunning{Marin Bougeret and  Ignasi Sau}
\titlerunning{How much does a treedepth modulator help to obtain polynomial kernels?}

\institute{LIRMM, Universit\'e de Montpellier, Montpellier, France.\\
\email{bougeret@lirmm.fr} \and
CNRS, LIRMM, Universit\'e de Montpellier, Montpellier, France, and\\Departamento de Matem\'atica, UFC, Fortaleza, Brazil.\\
\email{sau@lirmm.fr}}
\begin{document}
\maketitle

\begin{abstract}
In the last years, kernelization with structural parameters has been an active area of research within the field of parameterized complexity. As a relevant example, Gajarsk{\'y} et al.~[ESA 2013] proved that every graph problem satisfying a property called finite integer index admits a linear kernel on graphs of bounded expansion and an almost linear kernel on nowhere dense graphs, parameterized by the size of a $c$-treedepth modulator, which is a vertex set whose removal results in a graph of treedepth at most $c$, where $c \geq 1$ is a fixed integer. The authors left as further research to investigate this parameter on general graphs, and in particular to find problems that, while admitting polynomial kernels on sparse graphs, behave differently on general graphs.

In this article we answer this question by finding two very natural such problems: we prove that \textsc{Vertex Cover} admits a polynomial kernel on general graphs for any integer $c \geq 1$, and that \textsc{Dominating Set} does {\sl not} for any integer $c \geq 2$ even on degenerate graphs, unless $\text{NP} \subseteq \text{coNP}/\text{poly}$. For the positive result, we build on the techniques of Jansen and Bodlaender~[STACS 2011], and for the negative result we use a polynomial parameter transformation for $c\geq 3$ and an \textsc{or}-cross-composition for $c = 2$. As existing results imply that \textsc{Dominating Set} admits a polynomial kernel on degenerate graphs for $c = 1$, our result provides a dichotomy about the existence of polynomial kernels for \textsc{Dominating Set} on degenerate graphs with this parameter.

%We also study the following question: for problems admitting polynomial kernels on graphs of bounded expansion, are there smaller parameters for which the result still holds? While it is known that just the treedepth of the input graph does not allow for polynomial kernels, we consider as parameter the size $x$ of a vertex set whose removal results in a graph of treedepth at most $f(x)$, for a function $f$. We observe that \textsc{Dominating Set} does {\sl not} admit polynomial kernels on graphs of bounded expansion for $f(x) = \Omega (\log x)$ and that it does for $f(x) = O(\log \log \log x)$.

\end{abstract}
%-------------------------------------------------------------------------------------------------------------------------------

%\ig{remove references to appendices}
%\ig{fix definitions of problems}

 \section{Introduction}
\label{sec:intro}

\noindent \textbf{Motivation}. There is a whole area of parameterized algorithms and kernelization investigating the \textit{complexity ecology} (see for example~\cite{niedermeier2010reflections}),
where the objective is to consider a \emph{structural} parameter measuring how ``complex'' is the input, rather than the size of the solution.
For instance, parameterizing a problem by the treewidth of its input graph has been a great success for {\sf FPT} algorithms, triggered by Courcelle's theorem~\cite{courcelle1990monadic} stating that any problem expressible in MSO logic is {\sf FPT} parameterized by treewidth. However, the situation is not as  good for kernelization, as many problems do not admit polynomial kernels when
    parameterized by treewidth unless $\text{NP} \subseteq \text{coNP}/\text{poly}$~\cite{bodlaender2009problems}.

%(cf.~\cite{bodlaender2009problems} for \IS or \DS).

Of fundamental importance within structural parameters are parameters measuring the so-called ``distance from triviality'' of the input graphs,  like the size of a vertex cover (distance to an independent set) or of a feedback vertex set (distance to a forest). Unlike treewidth, these parameters may lead to both positive and negative results for polynomial kernelization. An elegant way to generalize these parameters is to consider a parameter allowing to {\sl quantify} the triviality of the resulting instance, measured in terms of its treewidth. More precisely, for a positive integer $c$, a  \emph{$c$-treewidth modulator} of a graph $G$ is a set of vertices $X$ such that the treewidth of $G-X$ is at most $c$. Note that for $c=0$ (resp. $c=1$), a $c$-treewidth modulator corresponds to a vertex cover (resp. feedback vertex set).

Treewidth modulators have been extensively studied in kernelization, especially on classes of {\sl sparse} graphs, where they have been at the heart of the recent developments of meta-theorems for obtaining linear and polynomial kernels on graphs on surfaces~\cite{BFL+09}, minor-free graphs~\cite{FLST10}, and topological-minor-free graphs~\cite{garnero2015explicit,KLPRRSS16}, all based in a generic technique known as \emph{protrusion replacement}. However, as observed in~\cite{gajarsky2013kernelization,KLPRRSS16}, if one tries to move further in the families of sparse graphs by considering, for instance, graphs of bounded expansion, for several natural problems such as \textsc{Treewidth-$t$ Vertex Deletion} (minimizing the number of vertices to be removed to get a graph of treewidth at most $t$), parameterizing by a treewidth modulator is as hard as on general graphs.

This observation led Gajarsk{\'y} et al.~\cite{gajarsky2013kernelization} to consider another type of modulators, namely \emph{$c$-treedepth modulators} (defined analogously to $c$-treewidth modulators), where treedepth is a graph invariant --~which we define in Section~\ref{sec:prelim}~-- that plays a crucial structural role on graphs of bounded expansion and nowhere dense graphs~\cite{sparsity}. Gajarsk{\'y} et al.~\cite{gajarsky2013kernelization}  proved that any graph problem satisfying a property called \emph{finite integer index}  admits a linear
kernel on graphs of bounded expansion and an almost linear kernel on nowhere dense graphs when parameterized by the size of a $c$-treedepth modulator. Shortly afterwards this result was obtained, the authors asked~\cite{worker13} to investigate this parameter on {\sl general graphs}, namely to find natural problems that admit and that do not admit polynomial kernels parameterized by the size of a $c$-treedepth modulator. More precisely, are there natural problems $\Pi_1$ and $\Pi_2$ fitting into the framework of~\cite{gajarsky2013kernelization} such that $\Pi_1/c$-\tdmod admits a polynomial kernel on general graphs, but $\Pi_2/c$-\tdmod does not?\footnote{As defined in Section~\ref{sec:prelim}, ``$/c$-\tdmod'' means ``parameterized by the size of a $c$-treedepth modulator''.}

% The starting point of this work is~\cite{gajarsky2013kernelization},

%corresponding using our notations to the size of a $c$-treewidth (or $c$-treedepth) modulator.
%These parameters generalize the previous ones as \VC corresponds to a $0$-treewidth modulator, and FVS to a $1$-treewidth modulator.

%Another example is given by the \textsc{Feedback Vertex Set} problem, studied by Jansen \emph{et al}.~\cite{jansen2014parameter}  under several parameterizations.

%The starting point of this work is~\cite{gajarsky2013kernelization}, where the authors prove that any graph problem satisfying a property called finite integer index  admits a linear
%kernel in BE and an almost linear kernel in ND when parameterized by $c$-\tdmod (this result is stated formally in Theorem~\ref{thm:sparse} in Appendix C). This theorem generalizes several \emph{meta-kernelization} results (where the kernelization holds for any problem
%satisfying a given property) that were established for more restrictive sparse graph classes like planar graphs or minor-free graphs. We refer the reader to the introduction of~\cite{gajarsky2013kernelization} for more
%details on the history of meta-kernelization and the motivation of considering a $c$-treedepth modulator instead of a $c$-treewidth modulator.

%Since this important result, investigating $c$-treedepth modulator as a parameter is now a well established research direction which was for example mentioned in the open problems of~\cite{worker13},
%with in particular the question of the existence of a natural problem admitting a polynomial kernel on general graphs.

\vspace{.25cm}
\noindent \textbf{Our results}. In this article we answer the above question by proving that \textsc{Vertex Cover} and \textsc{Dominating Set} are such problems $\Pi_1$ and $\Pi_2$, respectively. Let us now elaborate a bit more on our results, the techniques we use to prove them, and how do they compare to previous work in the area (see the preliminaries of Section~\ref{sec:prelim} for any undefined terminology).

%In this work we examine the question of improving the result of~\cite{gajarsky2013kernelization} (which can be summarized as the existence of a polynomial kernel by $c$-\tdmod on ND for a large class of problems).  More precisely, we consider the two following directions of improvement: % for the two classical problems
%\begin{itemize}
%\item generalizing the result to larger graph classes: do all problems of~\cite{gajarsky2013kernelization} behave in the same way? More precisely, are there problems $\Pi_1$ and $\Pi_2$ fitting into the framework of~\cite{gajarsky2013kernelization} such that $\Pi_1/c$-\tdmod admits a polynomial kernel on general graphs, but $\Pi_2/c$-\tdmod does not admit a polynomial kernel on another superclass of BE (than ND)?
%\item considering a smaller parameter: what happens (in BE or even smaller graph classes) for $f$-\tdmod with $f$ non-constant, or simply $td$?
%\end{itemize}
%We start providing answers to these general questions by considering two classical problems: \textsc{Vertex Cover} and \textsc{Dominating Set}, denoted \VC and \DS, respectively.

%Recall that $\Pi$/$c$-\tdmod denotes problem $\Pi$ parameterized by the size of a $c$-treedepth modulator.

Note first  that both \VC/$c$-\tdmod and \DS/$c$-\tdmod (where \DS stands for \textsc{Dominating Set})  are {\sf FPT} on general graphs, as they are {\sf FPT}
by treewidth~\cite{courcelle1990monadic}, which is a smaller parameter than $c$-\tdmod, as for any graph $G$ and any integer $c \geq 0$, it holds that $\tw(G) \leq \td(G) - 1 \leq$ c-$\tdmod(G) + c -1$. Thus, asking for polynomial kernels is a pertinent question.

In Section~\ref{sec:pos} we prove that \VC/$c$-\tdmod admits a polynomial kernel on general graphs. Our approach is based on the techniques introduced by Jansen and Bodlaender~\cite{jansen2011vertex} to prove that \VC/$1$-\twmod (or equivalently, \IS/\FVS) admits a polynomial kernel. As in~\cite{jansen2011vertex}, we in fact provide a polynomial kernel for \IS/$c$-\tdmod, which is easily seen to be equivalent to \VC/$c$-\tdmod. More precisely, we use three reduction rules inspired from the rules given in~\cite{jansen2011vertex}, and we present a recursive algorithm that, starting from a $c$-treedepth modulator, constructs an appropriate $(c-1)$-treedepth modulator and calls itself inductively. The kernel obtained in this manner has $x^{2^{\O(c^2)}}$ vertices, where $x$ is the size of the $c$-treedepth modulator. This result completes the following panorama of structural parameterization for \textsc{Vertex Cover}, which has been a real testbed for structural parameterizations in the last years:
\begin{itemize}
\item[$\bullet$] \VC/$1$-\twmod (or equivalently, \VC/\FVS) admits a polynomial kernel~\cite{jansen2011vertex}.
\item[$\bullet$] \VC/$c$-\twmod for $c \ge 2$ does not admit a polynomial kernel unless $\text{NP} \subseteq \text{coNP}/\text{poly}$~\cite{cygan2014hardness}.
\item[$\bullet$] \VC/$2$-\degmod (distance to a graph of maximum degree $2$) and \VC/$c$-\CVD (distance to a disjoint collection of cliques of size at most $c$) admit a polynomial kernel~\cite{majumdar2015kernels}.   Note that our result generalizes the latter kernel, as a disjoint collection of cliques of size at most $c$ is a particular case of a graph having treedepth at most $c$.
    %, where the obtained equivalent hypergraph has at most $\O(k^c)$ vertices ($k$ is the size of the modulator), and cannot have $\O(k^{c-\epsilon})$ vertices unless $\text{NP} \subseteq \text{coNP}/\text{poly}$.
\item[$\bullet$] \VC/\pfm (distance to a pseudoforest) admits a polynomial kernel~\cite{DBLP:journals/corr/FominS16}.
\end{itemize}

%Notice also that our result generalizes the kernel for \VC/$c$-\CVD of~\cite{majumdar2015kernels}, as a disjoint collection of cliques of size at most $c$ is a particular case of a graph having treedepth at most $c$,

%and answers the open question of~\cite{worker13}.
%\fixme{Moreover, the lower bound of~\cite{majumdar2015kernels} also applies, implying that...}

%% Param par tdmod : positiv result.
%% In Sec .. we consider \VC/c-\tdmod (or equivalently \IS/c-\tdmod).
%% Related
%% - \VC/c-\tdmod almost lin kernel in ND \cite{metathm}
%% - \VC/1-\twmod : poly kernel \cite{bart \VC kerne revisited}
%% - \VC/c-\twmod : c >= 2 no poly kernel \cite{on the hardness on loosing widht}
%% - un peu plus éloigné mais meme genre: OCT/c-\twmod inter bip :  kernel  \cite{jansen kernel oct}
%% - \VC/max deg 2 modulator and \VC/pseudo forest modulator : poly kernel dans \cite{AAJOUTER stromme}
%% -> we provide \VC/c-\tdmod poly kernel in general graphs: kernel size $X^..$
%% \fixme{une phrase qui se la pète genre "c'est premier exemple naturel qui admet pour tout c-\tdmod dans graphes généraux ?", ou du moins tq pour c-\twmod ya pas}

In Section~\ref{sec:neg1} we turn to negative results for \textsc{Dominating Set}.
We provide a characterization, according to the value of $c$, of the existence of polynomial kernels for \DS/$c$-\tdmod on degenerate graphs. Indeed, using the results of Philip et al.~\cite{philip2009solving} it is almost immediate to prove that \DS/$1$-\tdmod (or equivalently, \DS/\VC) admits a polynomial kernel on degenerate graphs. For $c \geq 3$, we rule out the existence of polynomial kernels for \DS/$c$-\tdmod on $2$-degenerate graphs by a simple  polynomial parameter transformation  from \DS/$1$-\tdmod on general graphs, which does not admit polynomial kernels unless $\text{NP} \subseteq \text{coNP}/\text{poly}$~\cite{dom2014kernelization}. The remaining case, namely \DS/$2$-\tdmod, turns out to be more interesting, and we rule out the existence of polynomial kernels on 4-degenerate graphs by providing an \textsc{or}-cross-composition from $3$-\textsc{Sat}. This dichotomy for the existence of polynomial kernels for \DS/$c$-\tdmod on degenerate graphs is to be compared with the dichotomy for \VC/$c$-\twmod on general graphs discussed above~\cite{jansen2011vertex,cygan2014hardness}.

%We also consider a smaller parameter and show that \DS/log-\tdmod does not admit a polynomial kernel on BE unless $\text{NP} \subseteq \text{coNP}/\text{poly}$. \ig{the positive result holds with the following definition of modulator: the parameter is the size $x$ of a set of vertices such that the removal of this set results in a graph of treedepth at most $f(x)$}

%Finally, in  we consider parameterization by $\td$ (which is another type of parameter not larger than $c$-\tdmod).

As mentioned before, it is  commonly admitted that almost no natural problem admits a polynomial kernel parameterized by $\tw$, or even with $\td$. However, to the best of our knowledge the only published negative results are those in~\cite{bodlaender2009problems}, where the
authors prove that \IS/\tw\ and \DS/\tw\ do not admit a polynomial kernel unless $\text{NP} \subseteq \text{coNP}/\text{poly}$. As this result only holds for general graphs, for the sake of completeness
we complete it in Section~\ref{sec:neg2}, by showing that a large majority of the problems considered in~\cite{gajarsky2013kernelization} having an almost linear kernel parameterized by $c$-\tdmod on nowhere dense graphs do not admit polynomial kernels parameterized by $\td$, even on planar graphs of bounded maximum degree.

\section{Preliminaries}
\label{sec:prelim}

\textbf{Graphs}. Unless explicitly mentioned, all graphs considered here are simple and undirected, and we refer the reader to~\cite{diestel2005graph} for any undefined notation.
Given a graph $G=(V,E)$ and $X \subseteq V$, we denote $N_X(v) = N(v) \cap X$, where $N(v) = \{u \in V \mid \{u,v\} \in E\}$.
We denote by $\alpha(G)$ the size of a maximum independent set of $G$.
For any function $f$ defined on any induced subgraph of a given graph $G$, given a subset of vertices $V'$ of $G$, we denote $f(V')=f(G[V'])$ (for example, $\alpha(V') = \alpha(G[V'])$).
For any integer $n$, we denote $[n]=\{i \in \mathbb{N}, 1 \le i \le n\}$.

For the following definitions related to treedepth, bounded expansion, and nowhere dense graph classes, we refer the reader to~\cite{sparsity} for more details,
and we only recall here some basic notations and facts.
The \emph{treedepth} of a graph $G$ (denoted $\td(G)$) is the minimum height of a rooted forest $F$ (called a \emph{treedepth decomposition}) such that $G$ is a subgraph of the closure of $F$, where the closure of a rooted tree is the graph obtained by adding an edge between any vertex and all its ancestors, and the height of a rooted tree is the number of vertices in a longest path from the root to a leaf.
Let $c \ge 1$ be an integer. A $c$-\emph{treedepth modulator} is a subset of vertices $X \subseteq V$ such that $\td(G[V\setminus X]) \le c$, and we denote by $c$-$\tdmod(G)$ the size of a smallest $c$-treedepth modulator of $G$. A $c$-treewidth modulator is defined in the same way.
Recall that as these parameters are greater than their associated measure (\ie $\tw(G) \le c$-$\twmod(G)+c$ and $\tw(G) \leq \td(G) \leq c$-$\tdmod(G)+c$, where $\tw(G)$ denotes the treewidth of $G$) the negative results for kernelization by treewidth and treedepth do not immediately apply, but
the positive {\sf FPT} results do.

%Thus, treewidth or treedepth modulator constitute natural candidates for polynomial kernelization.

Concerning graph classes, we recall that in the sparse graph hierarchy, graphs of bounded expansion (BE) and nowhere dense graphs (ND) are related to classic sparse families as follows (see~\cite{sparsity} for the definitions): planar graphs $\subseteq$ minor-free graphs $\subseteq$ BE $\subseteq$ ND.
Note also that the class of graphs of bounded degeneracy is a natural superclass of BE (intuitively, BE also requires the shallow minors to be degenerate), and is incomparable with ND.

We refer the reader to Appendix~\ref{ap:problems} for the definition and acronyms of problems considered in the paper, like \IS for the \textsc{independent set} problem.

\vspace{.25cm}
\noindent\textbf{Parameterized complexity}. We refer the reader to~\cite{downey2013fundamentals,CyganFKLMPPS15,FG06,Nie06} for more details on parameterized complexity and kernelization, and we recall here only some basic definitions, with special emphasis on tools for polynomial kernelization.
A \emph{parameterized problem} is a language $L \subseteq \Sigma^* \times \mathbb{N}$, for some finite alphabet $\Sigma$.  For an instance $I=(x,k) \in \Sigma^* \times \mathbb{N}$, $k$ is called the \emph{parameter}.
Given a classical (non-parameterized) decision problem $L_{c} \subseteq \Sigma^*$ and a function $\kappa: \Sigma^* \rightarrow \mathbb{N}$, we denote by
$L_{c}/\kappa = \{(x,\kappa(x)\} \mid x \in L_{c}\}$ the associated parameterized problem. For example, \IS/$c$-\tdmod denotes the \textsc{independent Set} problem parameterized by the size
of a $c$-treedepth modulator.

A parameterized problem is \emph{fixed-parameter tractable} ({\sf FPT}) if there exists an algorithm $A$, a computable function $f$, and a constant $c$ such that given an instance $I=(x,k)$,
$A$ (called an {\sf FPT} algorithm) correctly decides whether $I \in L$ in time bounded by \mbox{$f(k) \cdot |I|^c$}.
Given a computable function $g$, a \emph{kernelization algorithm} (or simply a \emph{kernel}) for a parameterized problem $L$ of \emph{size} $g$ is an algorithm $A$ that given any instance $I=(x,k)$ of $L$, runs in polynomial
time and returns an equivalent instance $I'=(x',k')$ (\ie $I'\in L$ if and only if $I\in L$) with $|I'|+k' \le g(k)$.
It is well-known that the existence of an {\sf FPT} algorithm is equivalent to the existence of a kernel (whose size may be exponential or larger), implying that problems admitting a polynomial kernel form a natural subclass of {\sf FPT}.

\vspace{.25cm}
\noindent\textbf{Some tools for (ruling out) polynomial kernelization}. Among the wide literature on polynomial kernelization, we only recall here the two following tools used in this paper: \emph{compositions} (used to prove that a parameterized problem does not admit a polynomial kernel unless $\text{NP} \subseteq \text{coNP}/\text{poly}$), and \emph{polynomial parameter transformations} from $L_1$ to $L_2$, denoted by $L_1 \le_{\PPT} L_2$ (used to prove that a polynomial kernel for $L_2$ implies a polynomial kernel for $L_1$).
%\section{Some tools for polynomial kernelization}\label{ap:tools}
\begin{definition}[\!\!\cite{bodlaender2009problems}]
An \emph{\OR-composition algorithm} for a parameterized problem $L \subseteq \Sigma^* \times \mathbb{N}$ is an algorithm that
\begin{itemize}
\item[$\bullet$] receives as input a sequence $((x_1,k),\dots,(x_t ,k))$, with $(x_i,k) \in  \Sigma^* \times \mathbb{N}$ for each $1 \le i \le t$,
\item[$\bullet$]  uses time polynomial in $\sum_{i=1}^t |x_i|+k$, and
\item[$\bullet$]  outputs $(y,k') \in \Sigma^* \times \mathbb{N}$ such that
 \begin{enumerate}
  \item $(y,k) \in L$ if and only if $(x_i ,k) \in L$ for some $1 \le i \le t$.
  \item $k'$ is polynomially bounded by a function of $k$.
 \end{enumerate}
\end{itemize}
We can similarly define an \AND-composition algorithm.
\end{definition}
A parameterized problem admitting a \OR-composition (resp. \AND-composition) is called an \emph{\OR-compositional} (resp. \AND-\emph{compositional}) problem. We first need the notion of \emph{polynomial equivalence relation}, defined as an equivalence relation $R$ on $\Sigma^*$ such that testing whether two strings $x,y$ are equivalent can be done in time polynomial in $|x|+|y|$, and such that $R$ restricted to the strings of size at most $n$ has at most $p(n)$ equivalence classes, for some polynomial $p$.

\begin{definition}[\!\!\cite{BodlaenderJK14}]\label{def:ORcross}
Let $L \subseteq \Sigma^*$ be a (classical) problem and $Q \subseteq \Sigma^* \times \mathbb{N}$
be a parameterized  problem. We say that $L$ \emph{\OR-cross-composes} (resp. \AND-\emph{cross-composes}) into $Q$ if there exists
a polynomial equivalence relation $R$ and an algorithm $A$, called the \emph{cross-composition}, satisfying the following conditions. The algorithm $A$ takes as
input a sequence of strings $x_1,\dots,x_t \in \Sigma^*$ that are equivalent with respect to $R$,
runs in time polynomial in $\sum_{i=1}^t|x_i|$, and outputs one instance $(y,k) \in \Sigma^* \times \mathbb{N}$ such that:
\begin{itemize}
\item[$\bullet$] $k \le p(\max_{i=1}^t |x_i| + \log(t))$ for some polynomial $p(\cdot)$, and
\item[$\bullet$] $(y, k) \in Q$ if and only if there exists at least one index $i$ such that $x_i \in L$ (resp. if for every $i \in [t], x_i \in L$).
\end{itemize}
\end{definition}

Compositions are a great tool to get negative results for kernelization.

\begin{theorem}[\!\!\cite{downey2013fundamentals,CyganFKLMPPS15}]\label{thm:compo}
Let $L$ be an \OR-compositional or an \AND-compositional parameterized problem whose derived classical problem is \NP-complete.
Then, $L$ does not admit a polynomial kernel unless \emph{$\text{NP} \subseteq \text{coNP}/\text{poly}$}.
\end{theorem}

 A \emph{polynomial compression} of a parameterized language
$Q \subseteq \Sigma^* \times \mathbb{N}$ into a language $R \subseteq \Sigma^*$ is an algorithm that takes as input
an instance $(x,k) \in \Sigma^* \times \mathbb{N}$, works in time polynomial in $|x|+k$, and returns
a string y such that
$|y| \leq p(k)$ for some polynomial $p$, and $y \in R$ if and only if $(x,k) \in Q$.

\begin{theorem}[\!\!\cite{downey2013fundamentals,CyganFKLMPPS15}]\label{thm:cross}
Assume that an \NP-hard language $L$ cross-composes into a
parameterized language $Q$. Then $Q$ does not admit a polynomial compression
unless \emph{$\text{NP} \subseteq \text{coNP}/\text{poly}$}.
\end{theorem}
Note that as a polynomial kernel is also a polynomial compression, the previous theorem also rules out the existence of a polynomial kernel.

\begin{definition}[\!\!\cite{bodlaender2008analysis}]\label{def:ppt}
Let $P$ and $Q$ be parameterized problems. We say that $P$ is \emph{polynomial
time and parameter} reducible to $Q$, written $P\le_{\PPT} Q$, if there exist a polynomial time
computable function $f:\Sigma^* \times \mathbb{N} \rightarrow \Sigma^* \times \mathbb{N}$ and a polynomial
$p$, such that for all $x \in \Sigma^*$ and $k \in \mathbb{N}$, if $f(x,k) = (x',k')$, then the following hold:
\begin{itemize}
\item[$\bullet$]  $(x,k) \in P$, if and only if $(x',k') \in Q$, and
\item[$\bullet$] $k' \le p(k)$.
\end{itemize}
We call $f$ a \emph{polynomial time and parameter transformation} from $P$ to $Q$ (\emph{PPT} for short).
\end{definition}

The following theorem can be used to obtain  either positive or negative results.
\begin{theorem}[\!\!\cite{bodlaender2008analysis}]\label{thm:ppt}
Let $P$ and $Q$ be parameterized problems, and suppose that $P_c$ and $Q_c$ are
the derived classical problems. Suppose that $Q_c$ is \NP-complete, and that $P_c \in$ \NP. Suppose
that $P \le_{\PPT} Q$. If $Q$ has a polynomial kernel, then $P$ has a polynomial kernel.
\end{theorem}

%Parameterized complexity allows a more detailed disinction between problems as a parameterized problem can be \NP-hard for \fixme{plus sur existence no pb XP pas FPT, et FPT mais pas poly kernel, }

\section{A polynomial kernel for \VC/$c$-\tdmod on general graphs}
\label{sec:pos}

In this section we prove that for any positive integer $c$, \VC/$c$-\tdmod admits a polynomial kernel on general graphs.
Recall that this was only known for \VC/$1$-\tdmod and \VC/$2$-\tdmod, as for $c=1$ this corresponds to the
standard parameterization and we can use the linear kernel of~\cite{abu2007crown},
and for $c=2$ we have $1$-\twmod $\le 2$-\tdmod (as a $1$-\twmod corresponds
to the distance to a forest, while $2$-\tdmod corresponds to the
distance to a star forest), and thus we can use the polynomial kernel
of~\cite{jansen2011vertex} for \VC/$1$-\twmod.
We also recall that we cannot expect to extend our result to
\VC/$c$-\twmod for any $c \geq 2$~\cite{cygan2014hardness}.

%, as \VC/$c$-\twmod is unlikely to admit a polynomial kernel for any $c \geq 2$~\cite{cygan2014hardness}.

As \VC/$c$-\tdmod and \IS/$c$-\tdmod are equivalent for this
parameterization (because any $n$-vertex graph has vertex cover of size at most $k$ if and only if it has an independent set of size at least $n-k$), we provide the result for \IS/$c$-\tdmod.
More specifically, in Subsection~\ref{subsec:annot} we provide a polynomial kernel for a-$c$-\tdmod-\IS, an annotated version of our problem defined below,
and in Subsection~\ref{subsec:normal} we derive a polynomial kernel for \IS/$c$-\tdmod.

%%%%%%%%%%%%%%%%%%%%% new version starts here
%In all this section, unless explicitely mentionned $(G,x)$ (where $G=(V,E)$, and $x$ is the size of a $c$ treedepth modulator of $G$) will
%denote the input graph of $\IS/c$-\tdmod. We denote by $X$ a $c$-\tdmod of $G$ (which is not given in the input), and $R=V \setminus X$.

\subsection{A polynomial kernel for a-$c$-\tdmod-\IS/$(|X|+|\H|)$}\label{subsec:annot}\label{subsec:kernel_annotated}

We will find a polynomial kernel for the following annotated version
of \IS on hypergraphs. Working with hypergraphs is useful because we will use a reduction rule identifying a subset $X'$ of the modulator that cannot be entirely contained in a solution; this will be modeled by adding a hyperedge on the vertex set $X'$.

%\fixme{rephrase, dire qu'on a une regle qui nous dit quun sous ensemble X' du moudaltor peut pas etre entierement dans une sol.
%Du coup on modelise ca en ajoutant lhyperarete X'}

%Intuitively, working on hypergraphs turns out
%to be helpful because hyperedges (which are allowed only within the
%modulator, and which will be introduced when applying the reduction
%rules) model the property that some set of vertices cannot be entirely
%included in an independent set.

%\igIMP{The name of the problem should be a-\IS-$c$-\tdmod/$(|X|+|\H|)$}
\vspace{.4cm}
\begin{boxedminipage}{.99\textwidth}
\textsc{Annotated $c$-treedepth modulator Independent Set} (a-$c$-\tdmod-\IS)\vspace{.1cm}

\begin{tabular}{ r l }
\textbf{~~~~Instance:} & $(G,X,k)$ where \\
 & ~\textbullet\ $G=(V,E,\H)$ is a hypergraph structured as follows: $V = X \uplus R$,\\
 & ~~ $E = E_{X,R} \uplus E_{R,R}$ is a set of edges where edges in $E_{A,B}$ have one \\
 & ~~  endpoint in $A$ and the other in $B$, and $\H \le 2^{X}$ is a set of \\
 & ~~  hyperedges where each $H \in \H$ is entirely contained in $X$.\\
 & ~\textbullet\ $X$ is a $c$-treedepth modulator (as $G[V\setminus X]$ is not a hypergraph, \\
 & ~~ its treedepth is correctly defined and we have $\td(V\setminus X) \le c$).\\
 & ~\textbullet\ $k$ is a positive integer.\vspace{.1cm}\\

\textbf{Question:} & Decide whether $\alpha(G) \ge k$ (an independent set in a hypergraph is a \\
 & subset of vertices that does not contain any hyperedge, corresponding  \\
 & here to a subset $S \subseteq V$ such that for every $h \in E \cup H$, $h \nsubseteq S$).\\
\end{tabular}
\end{boxedminipage}
\vspace{.4cm}

Throughout this subsection $I=(G,X,k)$ denotes the input of a-$c$-\tdmod-\IS with $G=(V,E,\H)$ and $V = X \uplus R$.
Note that $G[X]$ is a hypergraph and that $G[R]$ is a graph, and that the parameter we consider here is $|X|+|\H|$.
For any $X' \subseteq X$ and $R' \subseteq R$, observe that the notation $N_{R'}(X')$ is not ambiguous and denotes $\{v \in R' \mid \exists x \in X' \mbox{ with } \{x,v\} \in E \}$.
%Recall also that an independent set of $G$ is a subset $S \subseteq V$ such that for any $h \in E \cup H$, $h \nsubseteq S$.

We use the following definition that was introduced in~\cite{jansen2011vertex} for \VC/$1$-\twmod.
\begin{definition}[\!\!\cite{jansen2011vertex}]
Given $X' \subseteq X$ and $R' \subseteq R$, let $\conf_{R'}(X') =
\alpha(R')-\alpha(R'\setminus N_{R'}(X'))$ be the \emph{conflicts} induced by $X'$
on $R'$.% (recall that $\alpha(V')=\alpha(G[V'])$).
%\fixme{pkoi tu dis qu'on a pas défini "conflicts"? on le défini là justement}
\end{definition}

Intuitively, $\conf_{R'}(X')$ measures the loss in the size of a maximum independent set of
$R'$ due to $X'$. We extend the previous definition in the following way: for any $R' \subseteq R$ and any $Y' \subseteq R'$,
let $\conf_{R'}(Y') = \alpha(R')-\alpha(R'\setminus Y')$.
 Note that $\conf_{R'}(X')$ describes the impact of having $X'$ in the independent set, while $\conf_{R'}(Y')$ describes the impact of forbidding $Y'$ in the independent set.
 We can see that $\conf_{R'}(Y')=0$ is equivalent to the existence of an independent set $S^* \subseteq R'$ such that
$|S^*|=\alpha(R')$ and $S^* \cap Y' = \emptyset$.

\begin{lemma}\label{lemma:bsup_2c}
Let $R' \subseteq R$ be a connected component of $R$. %, where $\td(R')=\td(G[R'])$), and $Y' \subseteq R'$.
If $\conf_{R'}(Y')>0$, there exists $\bar{Y'} \subseteq Y'$ such that $\conf_{R'}(\bar{Y'})>0$ and $|\bar{Y'}| \le f(c)$ with $f(c)=2^c$.
\end{lemma}

\begin{proof}
%% Let us first show that proving the lemma when $\forall x \in X'$, $d_{R'}(x)=1$ is sufficient.
%% Let us consider $R'$ and $X'$ as stated in the lemma (with $d_{R'}(x)$ possibly greater than $1$).
%% For any $x \in X'$ with $N_{R'}(x)=\{v_i, i \in [l]\}$ we create vertices $\{y^x_i, i \in [l_x]\}$ with $N_{R'}(y^x_i)=\{v_i\}$,
%% and we define $Y' = \{y^x_i | x \in X', i \in [l_x]\}$. As we have $N_{R'}(Y')=N_{R'}(X')$ and $d_{R'}(y)=1$ for any $y \in Y'$ we can apply the lemma in our special case and
%% we know that there exists a set $\bar{Y'} \subseteq Y'$ such that $\conf_{R'}(\bar{Y'})>0$ and $|\bar{Y'}| \le 2^c$. As for any $\bar{Y'} \subseteq Y'$ there always exists $\bar{X'} \subseteq X'$ such that
%% $|\bar{X'}|=|\bar{Y'}|$ and $N_{R'}(\bar{X'}) \supseteq N_{R'}(\bar{Y'})$ (by taking $\bar{X'}=\{x \in X' | \exists i \in [l_x] \mbox{ such that } y^x_i \in \bar{Y'}\}$), we get the desired result.
%% We assume from now on that $\forall x \in X'$, $d_{R'}(x)=1$.
As it holds that $\td(R')\le c$, let us consider a treedepth decomposition of $R'$ with root $r$ and $t \ge 1$ subtrees rooted at the children of $t$, where $A_i$, $i \in [t]$ is the vertex set of subtree $i$.
We can partition $Y' = \bigcup_{i \in [t+1]}Y'_i$ with $Y'_i \subseteq A_i$ for $i \in [t]$, $Y'_{t+1} \subseteq \{r\}$, where the $Y'_i$'s are possibly empty.
We will prove the lemma by induction on $c$.
Observe that $\sum_{i \in [t]} \alpha(A_i) \le \alpha(R') \le 1+\sum_{i \in [t]} \alpha(A_i)$, and thus we distinguish two cases according to the value of $\alpha(R')$.

\begin{figure}[!htbp]
\begin{center}
\includegraphics[width=1.015\textwidth]{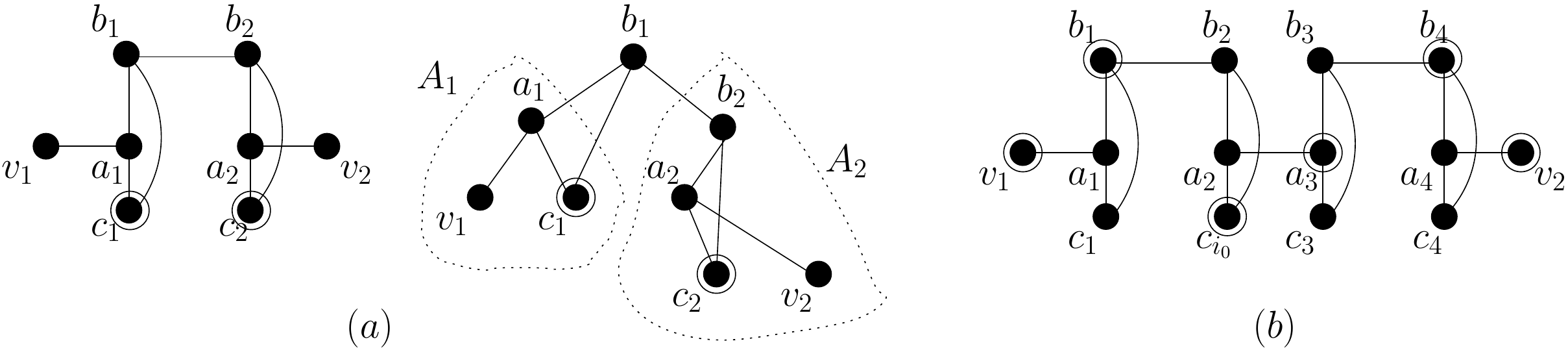}
\end{center}
\caption{(\emph{a}) Example of a graph $G[R']$ (left) with an associated treedepth decomposition (right) as used in Lemma~\ref{lemma:bsup_2c}, with $Y'=\{c_1,c_2\}$.
This case corresponds to one of the subcases of Case 2, as $\alpha(R')=\alpha(A_1)+\alpha(A_2)=4$, $\conf_{A_1}(Y'_1)>0$, $\conf_{A_2}(Y'_2)=0$. Moreover, $\mathbf{p_2}$ and $\mathbf{p'_2}$ are true, while $\mathbf{p_3}$ is false (but $\mathbf{p'_3}$ is true). (\emph{b}) Example for $t=2$ of the construction of Lemma~\ref{lemma:binf_2c}, where the circled vertices belong to $S$.}
\label{fig:bsup_2c}
\end{figure}

\medskip
\noindent{\bf Case 1:} $\alpha(R')=1+\sum_{i \in [t]} \alpha(A_i)$.
In this case every maximum independent set $S^*$ of $R'$ contains $r$. Hence for every $i \in [t]$,
$S^* \cap A_i$ is a maximum independent set in $A_i \setminus N_{A_i}(r)$, and thus $\alpha(A_i \setminus N_{A_i}(r)) = \alpha(A_i)$. Indeed, if we had $\alpha(A_i \setminus N_{A_i}(r)) < \alpha(A_i)$ for some $i$, then $|S^*|$ would be strictly smaller than $1+\sum_{i \in [t]} \alpha(A_i)$.

If $r \in Y'$ (\ie if $Y'_{t+1} \neq \emptyset$) then we can take $\bar{Y'}=\{r\}$ (as every optimal solution of $R'$ must contain $r$ we get
$\alpha(R' \setminus \{r\}) < \alpha(R')$, and $|\bar{Y'}|=1 \le 2^c$), and thus we suppose henceforth that $Y'_{t+1} = \emptyset$.

We claim that there exists $ i_0 \in [t]$ such that $\conf_{A_{i_0} \setminus N_{A_{i_0}}(r)}(Y'_{i_0}) > 0$. Indeed, otherwise we could define for every $i \in [t]$ an independent set $S_i \subseteq A_i \setminus N_{A_i}(r)$
with $|S_i|=\alpha(A_i \setminus N_{A_i}(r)) = \alpha(A_i)$ and $S_i \cap Y'_i = \emptyset$. Thus, $S^*=\{r\} \cup_{i \in [t]} S_i$ would be an independent set of size $\alpha(R')$, and as
$Y'_{t+1} = \emptyset$ we would have $S^* \cap Y' = \emptyset$, a contradiction to the hypothesis that $\conf_{R'}(Y')>0$.
Thus, there exists $ i_0 \in [t]$ such that $\conf_{A_{i_0} \setminus N_{A_{i_0}}(r)}(Y'_{i_0}) > 0$, and as $\td(A_{i_0} \setminus N_{A_{i_0}}(r)) < c$, by induction hypothesis there exists $\bar{Y'_{i_0}} \subseteq Y'_{i_0}$ such that $\conf_{A_{i_0} \setminus N_{A_{i_0}}(r)}(\bar{Y'_{i_0}}) > 0$ and $|\bar{Y'_{i_0}}| \le 2^{c-1}$. Let us verify that $\bar{Y'} = \bar{Y'_{i_0}}$ satisfies $\conf_{R'}(\bar{Y'})>0$. Let $S^*$ be an independent set of $R'$ with $S^* \cap \bar{Y'} = \emptyset$.
If $r \notin S^*$ then clearly $|S^*| < \alpha(R')$. Otherwise, $|S^*| = (\sum_{i \in [t]} |S^* \cap (A_i \setminus N_{A_i}(r))|)+1 \le \alpha(A_{i_0} \setminus N_{A_{i_0}}(r))-1 +
(\sum_{i \in [t], i \neq i_0} \alpha(A_i \setminus N_{A_i}(r)))+1 < \alpha(R')$.
%$S^* \cap A_{i_0}$ is an indepedent set of $A_{i_0} \setminus N_{A_{i_0}}(r)$, and as $S^* \cap N(\bar{Y'_{i_0}}) = \emptyset$ this implies that
%$|S^* \cap A_{i_0}| < \alpha(A_{i_0} \setminus N_{A_{i_0}}(r))$. Thus as for $i \neq i_0$, $|S^* \cap A_i| \le \alpha(A_i)$ we get $|S^*| < \alpha(R')$.

\medskip
\noindent{\bf Case 2:} $\alpha(R')=\sum_{i \in [t]} \alpha(A_i)$. In this case there exists $ i_0 \in [t]$ such that $\conf_{A_{i_0}}(Y'_{i_0})>0$. Indeed, otherwise we could define for every $i \in [t]$ an independent set
$S_i \subseteq A_i$ with $|S_i| = \alpha(A_i)$ and
$S_i \cap Y'_i = \emptyset$, and the existence of $S^* = \cup_{i \in [t]}S_i$ would be a contradiction to the hypothesis that $\conf_{R'}(Y')>0$.
Thus, by the induction hypothesis there exists $\bar{Y'_{i_0}} \subseteq Y'_{i_0}$ such that $\conf_{A_{i_0}}(\bar{Y'_{i_0}}) > 0$ and $|\bar{Y'_{i_0}}| \le 2^{c-1}$.

If $r \in Y'$ (\ie if $Y'_{t+1} \neq \emptyset$) then we can take $\bar{Y'}=\bar{Y'_{i_0}} \cup \{r\}$.
Let us verify that $\conf_{R'}(\bar{Y'})>0$. Let $S^*$ be an independent set of $R'$ with $S^* \cap \bar{Y'} = \emptyset$.
As $S^*$ cannot contain $r$ we have $|S^*| = \sum_{i \in [t]} |S^* \cap A_i| < \alpha(A_{i_0}) + \sum_{i \in [t], i \neq i_0} |S^* \cap A_i| = \alpha(R')$.
%By definition of $\bar{Y'_{i_0}}$ we have $|S^* \cap A_{i_0}| < \alpha(A_{i_0})$, and we also have that $r \notin S^*$. Thus as for $i \neq i_0$, $|S^* \cap A_i| \le \alpha(A_i)$ we get $|S^*| < \alpha(R')$.
Thus, we suppose from now on
\begin{equation*}
\text{Property } \mathbf{p_1}: Y'_{t+1} = \emptyset.
\end{equation*}
Note that in this case (when $\mathbf{p_1}$ is true) we cannot simply set $\bar{Y'} = \bar{Y'_{i_0}}$, as shown in the example depicted in Fig.~\ref{fig:bsup_2c}. Indeed, in this example we would have $\bar{Y'}=\bar{Y'_{i_0}}=\{c_1\}$,
however $\conf_{R'}(\{c_1\})=0$ as $S^*=\{b_1,v_1,c_2,v_2\}$ verifies $|S^*|=\alpha(R')$ and $S^* \cap \{c_1\} = \emptyset$.
%  with $t=1$ we have $\alpha(R')=\alpha(A_1)=2$.
%In this example the induction hypothesis only provides us a set $\bar{Y'_{1}}$ such that $\conf_{A_1}(\bar{Y'_1}) > 0$, and thus we could have $\bar{Y'_1} = \{1,2\}$. However,
%$\conf_{R'}(\bar{Y'_1}) = 0$ as we can take $S^* = \{r,r_1\}$ with $S^* \cap \bar{Y'_1} = \emptyset$.
%Thus, we need a more detailed analysis to handle case 2.

\medskip
\noindent{\em Properties related to $\alpha$.} Let us prove that we can always assume the following
\begin{equation*}
\text{Property}\ \mathbf{p_2}: \text{for every}\  i \neq i_0, \alpha(A_i \setminus N_{A_i}(r)) = \alpha(A_{i}).
\end{equation*}
Indeed, if $\mathbf{p_2}$ is not true, then there exists $i_1 \neq i_0$, $i_1 \in [t]$ such that $\alpha(A_{i_1} \setminus N_{A_{i_1}}(r)) < \alpha(A_{i_1})$, and we set $\bar{Y'}=\bar{Y'_{i_0}}$.
Let $S^*$ be an independent set of $R'$ with $S^* \cap \bar{Y'} = \emptyset$. If $r \notin S^*$ then as previously $|S^*| < \alpha(R')$, otherwise we get
$|S^*| \le \alpha(A_{i_0})-1+\alpha(A_{i_1})-1+ (\sum_{i \in [t],i \neq i_0, i \neq i_1} \alpha(A_i))+1< \alpha(R')$.
Thus, we now assume $\mathbf{p_2}$. Let us now prove the following
\begin{equation*}
\text{Property}\ \mathbf{p'_2}: \alpha(A_{i_0} \cup \{r\}) = \alpha(A_{i_0}).
\end{equation*}
By contradiction, suppose that there exists an independent set $S_1^*$ of $A_{i_0} \cup \{r\}$ containing $r$ such that $|S_1^*| = \alpha(A_{i_0})+1$. According to $\mathbf{p_2}$, for every $i \neq i_0$ there exists an independent set $S_i$ of $A_i \setminus N_{A_i}(r)$ of size $\alpha(A_i)$, and thus $\alpha(R') > \sum_{i \in [t]} \alpha(A_i)$, a contradiction.
Thus, we now assume $\mathbf{p'_2}$.

\medskip
\noindent{\em Properties related to $\conf_{A_i}(Y'_i)$.} Let us prove than we can assume the following
\begin{equation*}
\text{Property}\ \mathbf{p_3}: \text{for every}\  i \neq i_0,  \conf_{A_i \setminus N_{A_i}(r)}(Y'_i)=0.
\end{equation*}
Indeed, if $\mathbf{p_3}$ is not true we can get the desired result as follows. Let $i_1 \neq i_0$, $i_1 \in [t]$ such that $\conf_{A_{i_1} \setminus N_{A_{i_1}}(r)}(Y'_{i_1})>0$. We use the same arguments as in the previous paragraph and define $\bar{Y'}=  \bar{Y'_{i_0}} \cup  \bar{Y'_{i_1}}$. Note that $|\bar{Y'}| \le |\bar{Y'_{i_0}}|+|\bar{Y'_{i_1}}| \le 2^c$. Using the same notation, if $r \notin S^*$ then $|S^*| = (\sum_{i \in [t]} |S^* \cap A_i|) \le \alpha(A_{i_0})-1 + (\sum_{i \in [t], i \neq i_0} \alpha(A_i))< \alpha(R')$, and otherwise $|S^*| = (\sum_{i \in [t]} |S^* \cap (A_i \setminus N_{A_i}(r))|)+1 \le \alpha(A_{i_0})-1 + \alpha(A_{i_1})-1 + (\sum_{i \in [t], i \neq i_0, i \neq i_1} \alpha(A_i))+1 < \alpha(R')$.
Thus, we now assume $\mathbf{p_3}$. Note that $\mathbf{p_2}$ and $\mathbf{p_3}$ imply
\begin{equation*}
\text{Property}\ \mathbf{p'_3}: \text{for every}\   i \neq i_0,  \conf_{A_i}(Y'_i)=0.
\end{equation*}

%% If $\exists i_1 \neq i_0$, $i_1 \in [t]$ such that $\conf_{A_{i_1}}(Y'_{i_1})>0$. Using induction hypothesis $\exists \bar{Y'_{i_1}}$ such that $\conf_{A_{i_1}}(\bar{Y'_{i_1}}) > 0$ and $|\bar{Y'_{i_1}}| \le 2^{c-1}$.
%% We define $\bar{Y'}=  \bar{Y'_{i_0}} \cup  \bar{Y'_{i_1}}$. Let us verify that $\conf_{R'}(\bar{Y'})>0$. Let $S^*$ be an indepedent set of $R'$ with $S^* \cap N(\bar{Y'}) = \emptyset$.
%% We get $|S^*| \le (\sum_{i \in [t]} |S^* \cap A_i|)+1 \le (\sum_{i \in [t]} \alpha(A_i))-2+1< \alpha(R')$.
%% Thus, we suppose from now on the property $p_3$: $\forall i \neq i_0$,  $\conf_{A_i}(Y'_i)=0$.

%% Let us prove the following property $p_4$: $\conf_{A_{i_0} \cup \{r\}}(Y'_{i_0})>0$.
%% By contradiction, is $p_4$ was false there would be an independent set $S_{i_0}^*$ of $A_{i_0} \cup \{r\}$ (of size $\alpha(A_{i_0})$ according to $p'_2$) such that $S^*_{i_0} \cap Y'_{i_0} = \emptyset$.
%% According to $p_3$ there would exists for any $i \neq i_0$ an independent set $S^*_i$ of $A_i \setminus N_{A_i}(r)$ of size $\alpha(A_i \setminus N_{A_i}(r))=\alpha(A_i)$ (according to $p_2$) such
%% that $S^*_i \cap Y'_i = \emptyset$.
%% Thus, $S^*=\cup_{i \in [t]}S_i^*$ would be an independent set of size $\alpha(R')$ such that $S^* \cap Y'= \emptyset$, a contradiction.
%% Thus, we now assume $p_4$.
%% It remains two subcases to conclude the analysis.
%% \fixme{verif si encore besoni de p4}

\noindent \textbf{Case 2a}:  There does not exist a maximum independent set $S^*$  of $R'$ such that $r \in S^*$.
%$A_{i_0} \cup \{r\}$ (of size $\alpha(A_{i_0})$ according to $p'_2$) containing $r$.
In this case, we set $\bar{Y'} = \bar{Y'_{i_0}}$. Let us prove that $\conf_{R'}(\bar{Y'})>0$. Let $S^*$ be a maximum independent set of $R'$ with $S^* \cap \bar{Y'} = \emptyset$.
As $r \notin S^*$, we get  $|S^*| = \sum_{i \in [t]} |S^* \cap A_i| \le \alpha(A_{i_0})-1+  \sum_{i \in [t], i \neq i_0} \alpha(A_i) < \alpha(R')$.

\vspace{.2cm}
\noindent \textbf{Case 2b}: There exists a maximum independent set $S^*$ of $R'$ such that $r \in S^*$.
This implies that $\alpha(A_{i_0} \setminus N_{A_{i_0}}(r))=\alpha(A_{i_0})-1$.
Let us prove that $\conf_{A_{i_0} \setminus N_{A_{i_0}}(r)}(Y'_{i_0})>0$. If it was not the case, there would exist an independent set $S^*_{i_0}$ of $A_{i_0} \setminus N_{A_{i_0}}(r)$ of size
 $\alpha(A_{i_0} \setminus N_{A_{i_0}}(r))=\alpha(A_{i_0})-1$ such that $S^*_{i_0} \cap Y'_{i_0} = \emptyset$. By $\mathbf{p_3}$, there would exist, for every $i \neq i_0$, an independent set $S^*_i$ of $A_i \setminus N_{A_i}(r)$ of size $\alpha(A_i \setminus N_{A_i}(r))=\alpha(A_i)$ (by $\mathbf{p_2}$) such that $S^*_i \cap Y'_i = \emptyset$. Thus, $S^*=\{r\} \cup (\bigcup_{i \in [t]}S^*_i)$ would be an independent set of size $\alpha(R')$ such that $S^* \cap Y' = \emptyset$ (recall that by $\mathbf{p_1}$, $r \notin Y'$), a contradiction.
Thus, we know that both $\conf_{A_{i_0} \setminus N_{A_{i_0}}(r)}(Y'_{i_0})>0$ and $\conf_{A_{i_0}}(Y'_{i_0})>0$ (which was established at the beginning of Case 2).
Using twice the induction hypothesis  we get that there exists $\bar{Y'_{i_0}}^1 \subseteq Y'_{i_0}$ such that
$\conf_{A_{i_0} \setminus N_{A_{i_0}}(r)}(\bar{Y'_{i_0}}^1)>0$ and there exists $\bar{Y'_{i_0}}^2 \subseteq Y'_{i_0}$ such that
$\conf_{A_{i_0}}(\bar{Y'_{i_0}}^2)>0$, with both $|\bar{Y'_{i_0}}^1|$ and $|\bar{Y'_{i_0}}^2|$ bounded by $2^{c-1}$.
Thus, we set $\bar{Y'}=\bar{Y'_{i_0}}^1 \cup \bar{Y'_{i_0}}^2$. Let us verify that $\conf_{R'}(\bar{Y'})>0$. Let $S^*$ be an independent set of $R'$ with $S^* \cap \bar{Y'} = \emptyset$.
If $r \in S^*$, then $|S^*| = \sum_{i \in [t]} |S^* \cap (A_i \setminus N_{A_i}(r))| + 1 =  \alpha(A_{i_0} \setminus N_{A_{i_0}}(r))-1 + \sum_{i \in [t],i \neq i_0}\alpha(A_i) + 1= \alpha(A_{i_0})-2 + \sum_{i \in [t],i \neq i_0}\alpha(A_i) + 1 < \alpha(R')$.
Otherwise, $|S^*| = \sum_{i \in [t]} |S^* \cap A_i| =  \alpha(A_{i_0})-1 + \sum_{i \in [t],i \neq i_0}\alpha(A_i)  < \alpha(R')$. \end{proof}

A first lower bound on the function $f$ of Lemma~\ref{lemma:bsup_2c} can be obtained by considering a clique $R'$ on $c$ vertices (hence, with $\td(R')=c$) and
$Y'=R'$, as every $\bar{Y'} \subsetneq Y'$ satisfies $\conf_{R'}(\bar{Y'})=0$.
However, as shown in Lemma~\ref{lemma:binf_2c} below, we can even obtain an exponential lower bound, showing that the function $f(c) = 2^c$ of Lemma~\ref{lemma:bsup_2c} is almost tight.

%\fixme{pour la prop suviante on peut aussi prendre le gadget de tw1
%= echelle avec triangles à 2 pates , du coup laisser qd meem la prop suivante qui montre un peu mieux rec pourquoi faut sommer sur sous arbres ?}

\begin{lemma}\label{lemma:binf_2c}
There exists a constant $\lambda$ such that for every $c\ge \lambda$ there exists a graph $G=(R,E)$ and $Y \subseteq R$ such that
$\td(G)=c$,
$|Y| \ge 2^{c-3}$,
 $\conf_{R}(Y)>0$, and
for every $\bar{Y} \subsetneq Y$, $\conf_{R}(\bar{Y})=0$.
%\begin{itemize}
%\item $\td(G)=c$,
%\item $|Y| \ge 2^{c-3}$,
%\item $\conf_{R}(Y)>0$, and
%\item $\forall \bar{Y} \subsetneq Y$, $\conf_{R}(\bar{Y})=0$.
%\end{itemize}
\end{lemma}

\begin{proof}
Given an integer $t$, let $G=(R,E)$ with  $R = \bigcup_{i\in [2t]}\{a_i,b_i,c_i \} \cup \{v_1,v_2\}$ and
$E=\bigcup_{i \in [2t]}\{\{a_i,b_i\},\{c_i,a_i\},\{c_i,b_i\}\} \cup \bigcup_{i \in [t]}\{b_{2i-1},b_{2i}\} \cup \bigcup_{i \in [t-1]}\{a_{2i},a_{2i+1}\} \cup \{v_1,a_1\} \cup \{a_{2t},v_2\}$ ($G$ is a path on $2t+2$ vertices); see Fig.~\ref{fig:bsup_2c}(a) for an example for $t=1$.
We would like to point out that $R$ corresponds to the edge gadget of~\cite{cygan2014hardness}, except that we removed some edges (namely, $\{a_{2i-1},a_{2i}\}$) to lower its treedepth by a factor $2$.
Let $Y = C$.

We have $\alpha(R) = 2t+2$, and $\alpha(R \setminus Y) = \alpha(R)-1 < \alpha(R)$.
Let $\bar{Y} \subsetneq Y$, and let $i_0$ such that $c_{i_0} \notin \bar{Y}$. By symmetry, we can suppose that $i_0$ is even with $i_0=2i'_0$.
Let $S = \{v_1,v_2\} \cup \bigcup_{i \in [i'_0-1]}\{b_{2i-1},a_{2i}\} \cup \{b_{2i'_0-1},c_{2i'_0}\} \cup \bigcup_{i \in [i'_0+1,t]}\{a_{2i-1},b_{2i}\}$; see Fig.~\ref{fig:bsup_2c}(b) for an example for $t=2$. As $S$ is an independent set of size $\alpha(R)$
with $S \cap \bar{Y}=\emptyset$, we get that $\conf_{R}(\bar{Y})=0$.

Observe also that $\td(R) \le \log(t)+3$. Indeed, the treedepth of the initial path $P_{2t+2}$ on $2t+2$ vertices is at most $\lceil \log(2t+3) \rceil \le \log(t)+2$, for $t$ large enough.
Then, $\td(R) \le \td(P_{2t+2})+1$ as for every $i \in [2t]$, we can add to the treedepth decomposition of $P_{2t+2}$ a vertex $c_i$ as a new leaf attached to the lowermost vertex of $\{a_i,b_i\}$ in the decomposition.
\end{proof}
%observe that simpler construction vec une racine fails

\begin{observation}
Lemma~\ref{lemma:bsup_2c} was proven in~\cite{jansen2011vertex} when $R'$ is a forest and $|\bar{Y'}| \le 2$. Even if we already know that \IS/$2$-\twmod does not admit a
polynomial kernel unless $\text{NP} \subseteq \text{coNP}/\text{poly}$~\cite{cygan2014hardness}, it remains interesting to observe that, in
particular, this lemma becomes false for $2$-\twmod, as the graph of Lemma~\ref{lemma:binf_2c} has treewidth $2$.
This points out one crucial difference between $c$-treewidth and $c$-treedepth modulators.
\end{observation}

Let us now start the description of the kernel for a-$c$-\tdmod-\IS/$(|X|+|\H|)$.
Given an input $(G,X,k)$ of a-$c$-\tdmod-\IS, we define the following three rules.
Note that these rules and definitions (and the associated safeness proofs) correspond to Rules 1, 2, and 3 of~\cite{jansen2011vertex}, except that we now bound the sizes of the subsets by a function $f(c)$ instead of by $2$.

\begin{definition}
Given an input $(G,X,k)$ of a-$c$-\tdmod-\IS (with $\td(G[R]) \le c$ where $R = V \setminus X$), the \emph{chunks} of the input are defined by  $\X = \{X' \subseteq X \mid \mbox{ there is no } H \in \H \mbox{ such that } H \subseteq X', \mbox{ and } 0 < |X'| \le f(c)\}$, where $f(c)=2^c.$
\end{definition}
Intuitively, the chunks correspond to all possible small traces of an independent set of $G$ in $X$. We are now ready to define the first two rules.

\medskip
\noindent \textbf{Reduction Rule 1:} If there exists $u \in X$ such that $\conf_R(\{u\}) > |X|$, remove $u$ from the graph.

\medskip
\noindent \textbf{Reduction Rule 2:} If there exists $X' \in \X$ such that $\conf_R(X') > |X|$, add $X'$ to $\H$. %donc implicitement enlève des chunks

\medskip

%on définit les chunks après rule 1 et rule 2, sinon elles enlèves des chunks
\begin{lemma}\label{lemma:rule12}
Rule 1 and Rule 2 are safe: if $I=(G,X,k)$ is the original input of a-$c$-\tdmod-\IS and $I^1=(G^1,X^1,k)$ is the input after the application of Rule 1 or Rule 2, then $I$ and $I^1$ are equivalent.
\end{lemma}
\begin{proof}
Let us only prove the safeness of Rule 2, as Rule 1 corresponds to Rule 2 using $X' = \{u\}$. Indeed, adding hyperedge $\{u\}$ is equivalent to removing $u$, as by definition no independent set could contain $u$ anymore.
Let us prove that $\alpha(G) \ge k$ implies that $\alpha(G^1) \ge k$.
Let $S$ be an independent set of $G$ of size at least $k$. If $X' \nsubseteq S$ then $S$ is also an independent set of $G^1$ of size $k$.
Otherwise, let $S^1$ be an independent set of $R$ of size $\alpha(R)$.
Observe that $k \le |S| = |S \cap X|+|S \cap R| < |X|+(\alpha(R)-|X|)  = |S^1|$ (the strict inequality holds as $\conf_R(X')>|X|$), and we get the desired result.
\end{proof}

\medskip
\noindent \textbf{Reduction Rule 3:} If $R$ contains a connected component $R'$ such that for every $X' \in \X$, $\conf_{R'}(X')=0$, delete $R'$ from the graph and decrease $k$ by $\alpha(R')$.
\medskip

To prove that Rule 3 is safe we need the following lemma.
Recall that we say that $X' \subseteq X$ is an independent set if and only if there is no $H \in \H$ such that $H \subseteq X'$.
\begin{lemma}\label{lemma:prerule3}
Let $I=(G,X,k)$ be an instance of a-$c$-\tdmod-\IS.
Let $R'$ be a connected component of $R$.
If there exists an independent set $X' \subseteq X$ such that $\conf_{R'}(X') > 0$, then there exists $\bar{X'} \in \X$ such that
$\conf_{R'}(\bar{X'}) > 0$.% and $\bar{X'} \in \X$. %(\ie $|\bar{X'}| \le f(c)$.
\end{lemma}
\begin{proof}
Let $Y' = N_{R'}(X')$. As $\conf_{R'}(X') > 0$, $\conf_{R'}(Y') > 0$. By Lemma~\ref{lemma:bsup_2c}, there exists $\bar{Y'} \subseteq Y'$ such that
$\conf_{R'}(\bar{Y'})>0$ and $|\bar{Y'}| \le f(c)$. For every $y' \in \bar{Y'}$, there exists a vertex $g(y') \in X'$ such that $\{g(y'),y'\} \in E$,
and thus we define $\bar{X'}=\cup_{y' \in \bar{Y'}}g(y')$. As $\bar{X'} \subseteq X'$, $\bar{X'}$ is still an independent set, and $|\bar{X'}| \le |\bar{Y'}| \le f(c)$, we get that
$\bar{X'} \in \X$.
\end{proof}

\begin{lemma}\label{lemma:rule3}
Rule 3 is safe: if $I=(G,X,k)$ is the original input of a-$c$-\tdmod-\IS and $I'=(G',X',k')$ is the input after the application of Rule 3, then $I$ and $I'$ are equivalent.
\end{lemma}
\begin{proof}
$\alpha(G) \ge k \Rightarrow \alpha(G') \ge k'=k-\alpha(R')$ is straightforward, as if $S$ is an independent set of $G$ of size at least $k$ then $S \setminus R'$ is an independent
set of $G'$ of size at least $k-\alpha(R')$.

$\alpha(G) \ge k \Leftarrow \alpha(G') \ge k'=k-\alpha(R')$:
Let $S'$ be an independent set of $G'$ of size at least $k'$.
As Rule 3 applied, we know that for every $X_1 \subseteq \X$, $\conf_{R'}(X_1)=0$.
Using the contrapositive of Lemma~\ref{lemma:prerule3}, it follows that for every independent set $X_1 \subseteq X$, $\conf_{R'}(X_1) = 0$.
In particular we get that $X_S = S' \cap X$ verifies $\conf_{R'}(X_S)=0$. Thus, there exists an independent set $S_{R'}$ of $G[R']$ of size $\alpha(R')$ and such that $N_{R'}(X_S) \cap S_{R'} = \emptyset$,
and thus $S' \cup S_{R'}$ is an independent set of $G$ of size at least $k$.
\end{proof}

%fixme mini comment pour dire bound sur somme conflits aussi use dans cite, mais pas simplement comme ça car eux s'interesse pas (à verif) à juste borner nombre de cc.
\begin{lemma}\label{lemma:s}
Let $I=(G,X,k)$ be an instance of a-$c$-\tdmod-\IS, and let $s$ be the number of connected components of $R=V \setminus X$.
If none of Rule 1, Rule 2 and Rule 3 can be applied, then $s = \O(|X|^{f(c)+2})$, where $f$ is the function of Lemma~\ref{lemma:bsup_2c}.
\end{lemma}

\begin{proof}
First, as Rule 1 and Rule 2 cannot be applied, we have $\sigma=\sum_{X' \in \X} \conf_R(X') \le \sum_{i=1}^{f(c)} {|X| \choose i}|X|=\O(|X|^{f(c)+2})$.
On the other side, as Rule 3 cannot be applied, for every connected component $R' \subseteq R$ there exists $X' \in \X$ such that
$\conf_{R'}(X')>0$, and thus we have $\sigma \ge s$, implying the desired result.
\end{proof}

%% However, in our case (large $f(c)$), to make Rule 3 safe we should replace $2$ by $f(c)$ in the definition of chunks and in Rule 3.
%% In this way, the three rules would be safe. However, the argument to bound $s$ would no longer work as we would only have here
%% $\sum_{X' \mbox{ independent set of size at most $2$ }} \conf_R(X') \le |X|^3+|X|^2$ and $\sum_{X' \mbox{ independent set of size at most $f(c)$ }} \conf_R(X') \ge s$.
%% Thus, a natural way to fix this problem would be to modify Rule $2$ by considering any chunk $X'$ of size $f(c)$ instead of only $\{u,v\}$ of size $2$.
%% However, if for a given chunk $X'$ of size $f(c)$ we have $\conf_R(X') > |X|$, we now that we can prevent solutions to contain all vertices of $X'$, but this constraint cannot be
%% modelized by simply adding an edge inside $X'$. To handle these problems we define the following graph $Gtilde$ and its associated instance Itilde.

We are now ready to present  in Algorithm~\ref{algo:kernel} our polynomial kernel for a-$c$-\tdmod-\IS.

\SetKwInput{KwData}{Input}
%\SetEndCharOfAlgoLine{}
%\dontprintsemicolon

%\CommentSty{text}

\medskip
\begin{algorithm}[t]\label{algo:kernel}
\SetEndCharOfAlgoLine{}
 \KwData{$(I,c)$, where $I=(G,X,k)$ and $X$ is a $c$-treedepth modulator of $G$.} \vspace{.2cm}
% \KwResult{how to write algorithm with \LaTeX2e }

 \eIf{$c=0$}{
 \textbf{return} $X$.
 }{
 Apply Rule 1 exhaustively.

 \tcc*{this rule suppresses vertices of $X$}

 Apply Rule 2 exhaustively.

 \tcc*{this rule adds hyperedges of size at most $f(c)$ to $\H$}

 Define the set $\X$ of chunks.

 %\smallskip

 Apply Rule 3 exhaustively.

 \tcc*{this rule suppresses some connected components}

  \tcc*{of $R$ and decreases $k$ accordingly}

 Let $I_3=(G_3,X_3,k_3)$ be the obtained instance, where $G_3=(V_3,E_3)$ and $R_3=V_3\setminus X_3$.

 \smallskip

 For every connected component $R' \subseteq R_3$, compute an optimal treedepth decomposition of $R'$ with root $r_{R'}$.

 \smallskip

Let $X_r = \cup_{R' \subseteq R_3, R' \mbox{ {\scriptsize connected}}} \{r_{R'}\}$ be the set of roots.

\smallskip

Let $I'=(G'=(V_3,E',\H'),X',k_3)$ be defined as follows:

%$\ \ \ \ $ $V'=V_3$,

$\ \ \ \ $ $X' = X_3 \cup X_r$,

$\ \ \ \ $ $Z = \{e \in E_3 \mid e \cap X_r \neq \emptyset \mbox{ and } e \cap X_3 \neq \emptyset\}$,

$\ \ \ \ $ $E' = E_3 \setminus Z$,

$\ \ \ \ $ $\H' = \H_3 \cup Z$.

%$\ \ \ \ $ $k' = k_3$.

\tcc*{$I'$ corresponds to $I_3$ where we added $X_r$ to the modulator,
removed}

\tcc*{edges $Z$ from $E_3$, and added them as hyperedges of $X'$}

\tcc*{Note that $X'$ is now a $(c-1)$-treedepth modulator}
\textbf{return} $A(I',c-1)$.
 }
 %\While{not at end of this document}{
%  read current\;
%  \eIf{understand}{
%   go to next section  \;
%   current section becomes this one\;
%   }{
%   go back to the beginning of current section\;
%  }
% }
 \caption{A polynomial kernel for a-$c$-\tdmod-\IS/$(|X|+|\H|)$.}
\end{algorithm}

%\medskip
%\noindent
%$A(I,c)$: %(where $I=(G,X,k)$ and $X$ is a $c$-treedepth modulator)
%\begin{enumerate}
%\item If $c=0$, return $X$. Otherwise:
%\item While it is possible, apply Rule 1 (this rule suppresses vertices of $X$).
%\item While it is possible, apply Rule 2 (this rule adds hyperedges of size at most $f(c)$ to $\H$).
%\item Define the set $\X$ of chunks, and while it is possible, apply Rule 3 (this rule suppresses some connected components of $R$ and decreases $k$ accordingly).
%Let $I_3=(G_3,X_3,k_3)$ be the obtained instance, where $G_3=(V_3,E_3)$ and $R_3=V_3\setminus X_3$.
%\item for every connected component $R' \subseteq R_3$, compute an optimal treedepth decomposition of root $r_{R'}$.
%Let $X_r = \cup_{R' \subseteq R_3, R' \mbox{ {\scriptsize connected}}} \{r_{R'}\}$ be the set of roots.
%\item Let $I'=(G'=(V',E',H'),X',k')$ be defined as follows. Let $V'=V_3$, $X' = X_3 \cup X_r$, and
%$Z = \{e \in E_3 \mid e \cap X_r \neq \emptyset \mbox{ and } e \cap X_3 \neq \emptyset\}$.
%Let $E' = E_3 \setminus Z, \H' = \H_3 \cup Z$ and $k' = k_3$ ($I'$ corresponds to $I_3$ where we added $X_r$ to the modulator,
%and consequently removed edges $Z$ from $E_3$ and added them as hyperedges included in $X'$. Note that $X'$ is now a $(c-1)$-treedepth modulator).
%\item Return $A(I',c-1)$.
%\end{enumerate}
%\end{definition}

\begin{theorem}\label{thm:kernel-ann}
For every fixed integer $c \ge 0$, Algorithm~\ref{algo:kernel} is a polynomial kernel for \emph{a}-$c$-\tdmod-\IS/$(|X|+|\H|)$. More precisely, for every input $I=(G,X,k)$ (with $G = (V,E,\H)$, $R=V\setminus X$) where $X$ is a $c$-treedepth modulator,
Algorithm~\ref{algo:kernel} produces an equivalent instance $\tilde{I}=(\tilde{G},\tilde{X},\tilde{k})$
(with $\tilde{G} = (\tilde{V},\tilde{E},\tilde{H})$, $\tilde{R}=\tilde{V}\setminus \tilde{X}$) where $|\tilde{X}| = \O(|X|^{2^{(c+1)(c+2)/2}})$, $|\tilde{\H}| = |\H| + \O(|X|^{2^{(c+1)(c+2)/2}})$, and $\tilde{R} = \emptyset$.
\end{theorem}

\begin{proof}
Observe first that Algorithm~\ref{algo:kernel} is polynomial for fixed $c$. Indeed, computing $\conf_{R'}(X')$ is polynomial
(as $\tw(R') \le \td(R')$ and it is well-known that \IS/$\tw$ is {\sf FPT}~\cite{courcelle1990monadic}) and there are at most $\O(|X|^c)$ applications of Rules 1 and 2, and $\O(s|X|^c)$ applications of Rule 3.
Moreover, an optimal treedepth decomposition of each connected component can be computed in {\sf FPT} time parameterized by $c$, using~\cite{sparsity} or~\cite{ReidlRVS14}.
Let us prove the result by induction on $c$.
The result is trivially true for $c=0$. Let us suppose that the result holds for $c-1$ and prove it for $c$.
Observe that $X'$ is now a $(c-1)$-treedepth modulator, and thus we can apply the induction hypothesis on $A(I',c-1)$.
For every $\ell \in [3]$, let $I_\ell=(G_\ell,X_\ell,k_\ell)$ with $G_\ell=(V_\ell,E_\ell,\H_\ell)$ and $R_\ell=V_\ell \setminus X_\ell$ denote the instance after exhaustive application of Rule $\ell$, respectively.

\medskip
\noindent \emph{Equivalence of the output}.
By Lemma~\ref{lemma:rule12} and Lemma~\ref{lemma:rule3}, we know that Rules 1, 2, and 3 are safe, and thus that $I$ and $I_3$ are equivalent.
Note that $I_3$ is equivalent to $I'$ as the underlying input is the same (except that some vertices were added to the modulator).
As using induction hypothesis $A(I',c-1)$ outputs an instance $\tilde{I}$ equivalent to $I'$, we get the desired result.

\newpage
\noindent \emph{Size of the output}. We have
\begin{itemize}
\item[$\bullet$] $|X_1| \le |X|$, $|\H_1|=|\H|$.
\item[$\bullet$] $|X_2| = |X_1|$, $|\H_2| \le |\H_1|+|X_1|^{f(c)}$.
\item[$\bullet$] $|X_3| = |X_2|$, $|\H_3| = |\H_2|$. By Lemma~\ref{lemma:s}, $s$, the number of connected components of $R_3$, verifies $s = \O(|X_3|^{f(c)+2})$.
\item[$\bullet$] $|X'| \le |X_3| + s$, and $|\H'| \le |\H_3|+s|X_3|$.
\end{itemize}

Thus we get $|X'| = \O(|X|^{f(c)+2}) = \O(|X|^{2^{c+1}})$ and $|\H'| = |\H| + \O(|X|^{f(c)+3})$.
Using induction hypothesis we get that $|\tilde{X}| = \O(|X'|^{2^{c(c+1)/2}}) = \O(|X|^{2^{(c+1)(c+2)/2}})$, and that
$|\tilde{\H}| = |\H'|+ \O(|X'|^{2^{c(c+1)/2}}) =  |\H| + \O(|X|^{2^c+3})+\O(|X|^{2^{(c+1)(c+2)/2}}) = |H|+\O(|X|^{2^{(c+1)(c+2)/2}})$, as claimed.
\end{proof}

\subsection{Deducing a polynomial kernel for \IS/$c$-\tdmod}\label{subsec:normal}
Observe first that we can suppose that the modulator is given in the input, \ie that \IS/$c$-\tdmod $\le_{\PPT} $ $c$-\tdmod-\IS/$|X|$ ($\le_{\PPT}$ is defined in
Definition~\ref{def:ppt}).
Indeed, given an input $(G,x,k)$ of \IS/$c$-\tdmod (where $x$ denotes the size of a $c$-treedepth modulator), using the $2^c$-approximation algorithm of~\cite{gajarsky2013kernelization} for computing a $c$-treedepth modulator, wet get in polynomial time a set $X$ such that $|X| \le 2^c \cdot x$ and $\td(R) \le c$, where $R=V\setminus X$.

Observe also that \IS/$|X| \le_{\PPT} $ a-$c$-\tdmod-\IS/$(|X|+|\H|)$ using the same set $X$ and with $|\H| \le |X|^2$.
Now, as usual when using bikernels, we could claim that as \IS is Karp $\NP$-hard and as a-$c$-\tdmod-\IS is in $\NP$, there exists a polynomial reduction from a-$c$-\tdmod-\IS, implying the existence
of a polynomial kernel for \IS/c-\tdmod. However, let us make such a reduction explicit to provide an explicit bound on the size of the kernel.

\begin{lemma}\label{lem:explicit-reduction-annotated}
Let $I=(G,k)$ with $G=(X,\H)$ be an instance of a-$c$-\tdmod-\IS as produced by Theorem~\ref{thm:kernel-ann} (as $R= \emptyset$ the set of vertices is reduced to $X$, and $\H$ is a set of hyperedges on $X$).
We can build in polynomial time an equivalent instance $I'=(G',k')$ of \IS with $G'=(V',E')$ where $|V'| \le \O(|X| \cdot |\H|)$.
\end{lemma}
\begin{proof}
Let $n = |X|$, $X = \{v_i \mid i \in [n]\}$ and $m = |\H|$.
We refer the reader to Fig.~\ref{fig:reduc} for an example of the construction of $G'$.
For every $i \in [n]$, we add to $G'$ the vertex gadget constituted of $V'_i = \{y^a_i,y^b_i,z_i\}$ and edges $\{z_i,y^a_i\}$ and $\{z_i,y^b_i,\}$.
Taking vertices $y^a_i$ and $y^b_i$ in a solution for $I'$ will correspond to taking $v_i$ in the corresponding solution of $I$.
For every $H \in \H$, we add to $G'$ the edge gadget $W'_H = \bigcup_{\ell \in [|H|]}W'^\ell_H$, where each $W'^\ell_H$ is an independent set of size $n$, and
 we add edges to make $G[W'_H]$ a complete $|H|$-partite graph with ${|H| \choose 2} n^2$ edges.
Finally, for every $H \in \H$, $H=\{v_{H_\ell} \mid \ell \in [|H|]\}$ and $\ell \in [|H|]$, we add edges to make $G[W'^\ell_H \cup \{y^a_{H_\ell},y^b_{H_\ell}\}]$ a complete bipartite graph with $2n$ edges.
Thus, the $\ell$-th ``column'' of $W'_H$ corresponds to the $\ell$-th vertex of $H$.
This completes the description of $G'$. Let $k'=n+k+n m$.

\begin{figure}[t]
\begin{center}\vspace{-.5cm}
\includegraphics[width=0.45\textwidth]{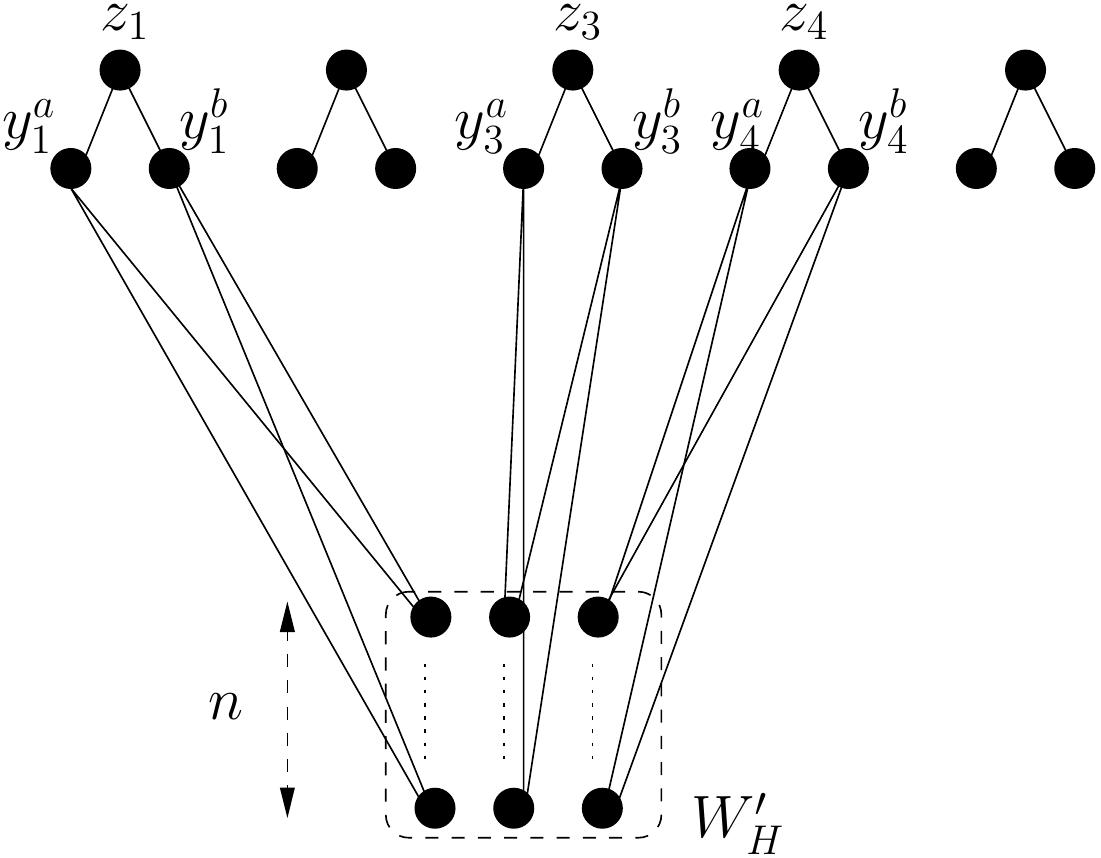}
\end{center}\vspace{-.15cm}
\caption{Example of the construction of $G'$ for $n=5$, $m=1$, and $H = \{1,3,4\}$.}
\label{fig:reduc}
\end{figure}

\noindent $\alpha(G) \ge k \Rightarrow \alpha(G') \ge k'$: Without loss of generality let $S=\{v_i \mid i \in [k]\}$ be an independent set of $G$.
For every $H=\{v_{H_\ell} \mid \ell \in [|H|]\}$ there exists $\ell$ such that $v_{H_\ell} \notin S$.
We define $S'=\bigcup_{i \in [k]}\{y^a_i,y^b_i\} \cup \bigcup_{i \in [n]\setminus [k]}\{z_i\} \cup \bigcup_{H \in \H}W'^{H_\ell}_H$.

\medskip
\noindent
$\alpha(G) \ge k \Leftarrow \alpha(G') \ge k'$: Let $S'$ be an independent set of $G'$ of size at least $k'$.
We can always assume that for every $H \in \H$ and $\ell \in [|H|]$, if $W'^\ell_H \cap S' \neq \emptyset$ then $W'^\ell_H \subseteq S'$
(as all vertices of $W'^\ell_H$ have the same neighborhood, we can safely add $W'^\ell_H$).
Note that there cannot exist $H \in \H$ such that $W'_H \cap S' = \emptyset$. Indeed, otherwise $|S'| \le n(m-1)+2n < k'$.

Thus, for every $H=\{v_{H_\ell} \mid \ell \in [|H|]\}$ there exists (a unique) $\ell_H \in [|H|]$ such that $W'^{\ell_H}_H \subseteq S'$.
As there remain $n+k$ vertices to take in the $V'_i$'s, and as $y^a_i$ and $y^b_i$ have the same neighborhood, we get that, without loss of generality,
for every $i \in [k]$ we have $\{y^a_{i},y^b_{i}\} \subseteq S'$, and for every $i \in [n]\setminus [k]$ we have $z_i \in S'$. Thus, we define $S = \{v_i \mid i \in [k]\}$.
Let us verify that $S$ is an independent set. Let $H \in \H$. As $W'^{\ell_H}_H \subseteq S'$ and $S'$ is an independent set, we deduce that there are no edges in $G'$
between $W'^{\ell_H}_H$ and any of the $\{y^a_{i},y^b_{i}\}$ for $i \in [k]$, implying that $v_{H_{\ell_H}} \notin S$, and therefore $S$ is indeed an independent set.
\end{proof}

Putting pieces together we get the main theorem of this section, whose proof is now immediate.

\begin{theorem}\label{thm:kernel-VC}
 For every integer $c \geq 1$, \IS/$c$-\tdmod (or equivalently, \VC/$c$-\tdmod) admits a polynomial kernel on general graphs with $\O\left(x^{2^{\frac{1}{2}(c+1)(c+2)+1}}\right)$ vertices, where $x$ is the size of a $c$-treedepth modulator.
\end{theorem}

%% fixme dire que on suppose que input = G,k,X

%% Our objective is no
%% preciser ce que mean kernel poly ici ,
%% -\IS/X <=ppt ann-mod-\IS-c/X+H
%% -ann-mod-\IS-c poly kernel par X+H
%%   aprs |G'|=|V'|+|\H'| <= poly(|X|+|H|) (pas juste poly X car si H l'était déjà pas c'est impossible,  et bien de dire car si H explose exponentiellement en X apres on pourra pas catalyser en restant poly
%% -ann-mod-\IS-c <=karp \IS/X
%Notice that $mod-\IS/|X| \le_{PPT} ann-mod-\IS-c/(|X|+|\H|)$ (by simply adding an hyperedge of size $2$ for each edge entirely contained in $X$).

\section{Excluding polynomial kernels for  \DS/$c$-\tdmod on degenerate graphs}
\label{sec:neg1}

Given a graph $G$, we define $G^{c\text{-sub}}$  as the graph obtained from $G$ by subdividing each edge $c$
times. In other words, we add a set $X_e=\{x_e^\ell \mid \ell \in [c]\}$ of $c$ vertices of degree $2$ for every edge
$e \in E$ of $G$.

\begin{observation}\label{obs:sub}
For every $c \ge 0$ and every $k \ge 0$, $G$ has a dominating set of size
$k$ if and only if \emph{$G^{3c\text{-sub}}$} has a dominating set of size $k+mc$, where $m$ is the number of edges of $G$.
\end{observation}
\begin{proof}
Indeed, if $S$ is a dominating set of $G$ of size $k$, then we construct
a dominating set of $G^{3c\text{-sub}}$ of size $k+mc$ by taking $S$ and the
following vertices. For every $e = \{u,v\}$ with $u \in S$ and $v \notin
S$ we take $\{x_e^{3\ell} \mid 1 \le \ell \le c \}$ (we add to the dominating set every third vertex in the set $X_e$ starting from $u$).
Otherwise (if both or none of $\{u,v\}$ belong to $S$) we take $\{x_e^{3\ell+2} \mid 0 \le \ell \le c-1 \}$.
The other direction is also true, as every solution must include at least
$c$ vertices in each $X_e$, and every solution can be modified so
that it does not include more than $c$ vertices in each $X_e$. Thus, the
$k$ vertices of a solution of $G^{3c\text{-sub}}$ corresponding to original
vertices of $G$ form a dominating set of $G$.
\end{proof}
%\fixme{add figure pour $c=1$?}

Let us start with the following proposition, which follows from
existing negative results for \textsc{Dominating Set} parameterized by the size of a vertex cover~\cite{dom2014kernelization}.

\begin{proposition}\label{prop:ds/tdmod}
\DS/$c$-\tdmod does not admit a polynomial kernel on $2$-degenerate graphs for every $c \ge 3$,  unless \emph{$\text{NP} \subseteq \text{coNP}/\text{poly}$}.
\end{proposition}

\begin{proof}
Let us prove that \DS/\VC $\le_{\PPT}$ \textsc{DS}$_{\C_{\text{dege}}}$/$3$-\tdmod, where $\C_{\text{dege}}$ is the class of $2$-degenerate graphs. As \DS/\VC
(and even \DS/$k$+\VC) does not admit a polynomial kernel unless $\text{NP} \subseteq \text{coNP}/\text{poly}$  \cite{dom2014kernelization}, we will get the desired
result using Theorem~\ref{thm:ppt}.
Let $(G,k)$ be an instance of \DS/\VC with  $G=(V,E)$ and $m=|E|$. We
define $G'=G^{3\text{-sub}}$, and let $V'$ be the set of vertices of $G'$. By Observation~\ref{obs:sub}, $G$ has a dominating set of
size $k$ if and only if $G'$ has a dominating set of size
$k+m$. Moreover, it is clear that $G'$ is $2$-degenerate. Finally, every vertex cover $X$ of $G$ is a
$3$-treedepth modulator of $G'$. Indeed, in $G'[V'\setminus X]$, to each edge $e
\in E$ entirely contained in $X$ corresponds in $G'[V'\setminus X]$ an
isolated $P_3$, and to each $v \in V \setminus X $ corresponds in
$G'[V' \setminus X]$ a spider (that is, a tree with only one vertex of degree more than two) rooted at $v$ of height $4$ with
$x\ge 1$ leaves. Thus, $G'[V'\setminus X]$ is a disjoint collection of
$P_3$'s and spiders of height $4$, both having treedepth at most $3$.
As $3$-\tdmod$(G') \le {\sf vc}(G)$ (the size of a minimum vertex cover of $G$), this is a PPT reduction and we get
the desired result. We can get the same result for
\textsc{DS}$_{\C_{\text{dege}}}$/$c$-\tdmod for $c \ge 4$ by subdividing $3f(c)$ times each
edge of $G$, for an appropriate function $f$.\end{proof}

\begin{observation}\label{obs:ds/vc}
\DS/$1$-\tdmod (or equivalently \DS/\VC) admits a polynomial kernel on degenerate graphs.
Indeed, given an instance $(G,k)$ of \DS/\VC, we compute in polynomial time a $2$-approximate vertex cover
$X$ of $G$. If $|X| \le k$ then we output a trivial \textsc{Yes}-instance,
otherwise VC$(G) \ge \frac{k}{2}$ and we can apply the polynomial
kernel for \DS/$k$ on degenerate graphs of Philip et al.~\cite{philip2009solving}.
%\fixme{je pense qu'il faut faire la mini explication et qu'on peut pas
%  juste dire \DS <= \VC et donc kernel par \DS/k implique kernelr pour
%  \DS/\VC, car on ferait la confusion entre \DS/\DS et \DS/k (enfin bref
%  c'est toujours la discussion avec la def de param complexi style
%  flhume que je suis pas sur de comprendre : cest suppose etre pareil
%  \DS/\DS et \DS/k ? mais si yavait pas d'approx poly pour \VC?}
\end{observation}

Thus, by Proposition~\ref{prop:ds/tdmod} and Observation~\ref{obs:ds/vc}, the only remaining case for degenerate
graphs is \DS/$2$-\tdmod.
We would like to point out that the composition of~\cite{dom2014kernelization} for \DS/($k+$\VC) on general graphs cannot be easily
adapted to \DS/$2$-\tdmod on degenerate graphs, as for example subdividing each edge also leads to a result for \DS/$3$-\tdmod.
Thus, we treat the case \DS/$2$-\tdmod on degenerate graphs using an ad-hoc reduction.

%\fixme{k+\DS/2tdmod dans dege ou mieux pour dire que useless dadapter compo qui preserve k?}

\begin{theorem}\label{thm:DS}
\emph{\DS/$2$-\tdmod} does not admit a polynomial kernel on $4$-degenerate graphs unless \emph{$\text{NP} \subseteq \text{coNP}/\text{poly}$}.
\end{theorem}
\begin{proof}
We prove this result by using an \OR-cross-composition from $3$-\textsc{Sat} (see Definition~\ref{def:ORcross}).
We consider $t$ instances of $3$-\textsc{Sat}, where for every $i \in [t]$, instance $I^i$ has $m_i$ clauses $\{C^i_j \mid j \in [m_i]\}$ and $n_i$
variables $X^i=\{x^i_{\ell} \mid {\ell} \in [n_i]\}$,  each clause containing $3$ variables. We can choose the equivalence relation  of Definition~\ref{def:ORcross} such that for every $i \in [t]$, we have
$m_i=m$ and $n_i=n$.
% We denote $U^i =\{u^i_l, l \in [u]\}$,
%$W^i =\{w^i_l, l \in [w]\}$.

\begin{figure}[t]
\begin{center}\vspace{-.4cm}
\includegraphics[width=0.57\textwidth]{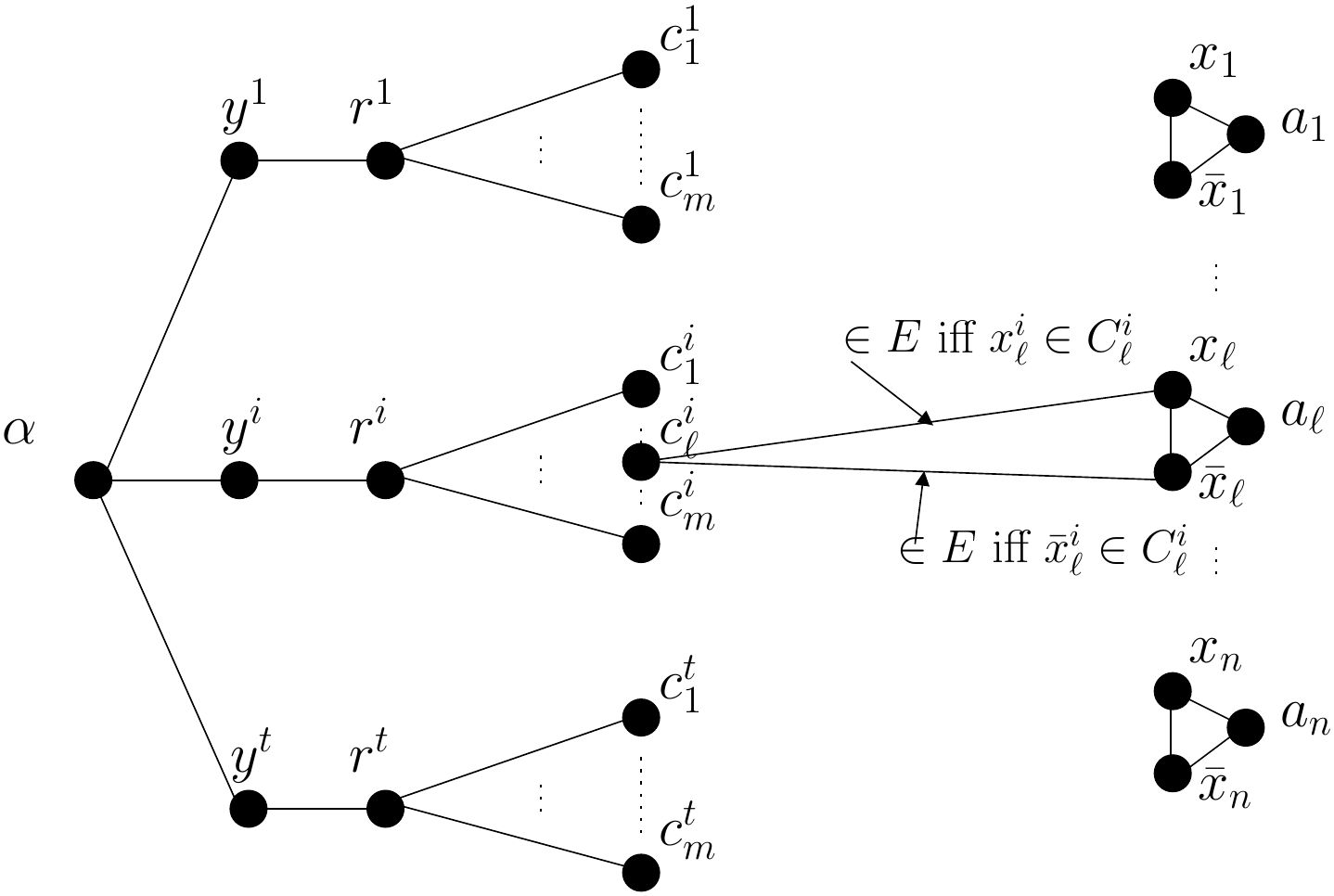}
\end{center}
\caption{Example of the \OR-cross-composition of Theorem~\ref{thm:DS}.}
\label{fig:cross}
\end{figure}

Let us now construct a graph $G=(V,E)$ as follows; see Fig.~\ref{fig:cross} for an illustration. We start by adding to $V$
the set of vertices $\X=\bigcup_{{\ell} \in [n]} \{\x_{\ell}, \bar{\x}_{\ell}\}$ (and
thus $|\X|=2n$) and $\C^i = \{\c^i_{\ell} \mid {\ell} \in [m]\}$ for every $i \in [t]$.
Let $\C = \bigcup_{i \in [t]} \C^i$.
For every $i \in [t]$, ${\ell} \in [n]$, $j \in [m]$, we set
$\{\x_{\ell},\c^i_j\}\in E^i$ (resp. $\{\bar{\x}_{\ell},\c^i_j\}\in E^i$) if and only if
$C^i_j$ contains $x^i_{\ell}$ (resp. $\bar{x}^i_{\ell}$).
We add to $E$ the set $\bigcup_{i \in [t]} E^i$. Then, we add to $V$ the set $\A=\{a_{\ell} \mid {\ell} \in [n]\}$, and create $n$
triangles by adding to $E$ edges $\{\x_{\ell},\bar{\x}_{\ell}\}$, $\{a_{\ell},\x_{\ell}\}$,
and $\{a_{\ell},\bar{\x}_{\ell}\}$ for every ${\ell} \in [n]$. Finally, we add to $V$ the set $\Y = \{\y^i \mid i \in [t]\}$, $\R = \{\r^i \mid i
\in [t]\}$, and a vertex $\alpha$. Then, for every $i \in [t]$, we add
to $E$ edges $\{\r^i,\c^i_{\ell} \}$ for every ${\ell} \in [m]$, edges
$\{\r^i,\y^i\}$, and edges $\{\y^i,\alpha\}$. This concludes the construction of $G$. To summarize, $G$ has $3n+t(m+2)+1$ vertices (vertices are
partitioned into $V=(\X \cup \A) \cup (\C \cup \Y \cup \R) \cup \{\alpha\}$) and, in particular, for every $i \in [t]$, $G[\{r^i\} \cup C^i \cup \y^i]$
is a star, and $G[\{\alpha\} \cup Y]$ is also a star.

\medskip
\noindent
\emph{The \OR-equivalence}.
Let us prove that there exists $i \in [t]$ such that $I^i$ is satisfiable
if and only if $G$ has a dominating set of size at most $k=n+t$.
Suppose first, without loss of generality, that $I^1$ is satisfiable, and let
$S_{\X} \subseteq \X$ be the set of $n$ literals corresponding to this
assignment (thus for every $\ell \in [n]$ we have $|S_\X \cap
\{\x_{\ell},\bar{\x}_{\ell}\}|=1$). Let $S = S_\X \cup \y^1 \cup (R \setminus
\{\r^1\})$. We have $|S|=n+t$, and  $S$ is a dominating
set of $G$ as
\begin{itemize}
\item[$\bullet$] $\X \cup \A$ is dominated by $S_\X$,
\item[$\bullet$] $\C^1$ is dominated by $S_\X$ as it corresponds to an assignment
  satisfying $I^1$, and for every $i \in [t], i \ge 2$,  $\C^i$ is dominated by $\r^i$,
\item[$\bullet$] $\y^1 \in S$, and for every $i \in [t], i \ge 2$,  $\y^i$ is dominated by $\r^i$,
\item[$\bullet$] $\r^1$ is dominated by $\y^1$, and for every $i \in [t], r \ge 2$,
  $\r^i \in S$, and
\item[$\bullet$] $\alpha$ is dominated by $\y^1$.
\end{itemize}

For the other direction, let $S = S_1 \cup S_2$, with $S_1 = S \cap (\X \cup \A)$, be a dominating set of $G$ of size at most $k=n+t$.
Without loss of generality, we can always suppose that $S_1 \subseteq \X$,
as if $\a_{\ell} \in S$ we can always remove $a_{\ell}$ from  $S$ and add (arbitrarily)
$\x_{\ell}$ or $\bar{\x}_{\ell}$.

Let us first prove that $|S_1|=n$.
Observe first that $|S_1| \ge n$ as dominating $A$ requires at
least $n$ vertices. Suppose now by contradiction that $|S_1| >
n$. Then, there would remain at most $t-1$ vertices to dominate $R$,
which is not possible. Note that we even have that for every $\ell \in [n]$, $|S_1 \cap\{\x_{\ell},\bar{\x}_{\ell}\}|=1$, as every $\a^{\ell}$ must be dominated and
$|S_2|=t$.

Let us now analyze $S_2$ (recall that, by definition, $S_2 \subseteq (\C \cup \Y \cup \R) \cup \{\alpha\}$).
We cannot have that for every $i \in [t]$, $|S_2 \cap
(\C^i\cup \r^i)| \ge 1$, as otherwise there would be no remaining vertex to dominate
$\alpha$.
Thus, there exists $i_0$ such that $|S_2 \cap (\C^{i_0}\cup
\r^{i_0})|= 0$. This implies that $\C^{i_0}$ is dominated by $S_1$.
As for every ${\ell} \in [n]$, $|S_1 \cap\{\x_{\ell},\bar{\x}_{\ell}\}|=1$, $S_1$
corresponds to a valid truth assignment that satisfies all the
$C^i_{\ell}$'s, ${\ell} \in [m]$, and the instance $I^{i_0}$ is satisfiable.

\medskip
\noindent
\emph{Size of the parameter}. Let $M = \X \cup \A \cup \{\alpha\}$. As $G[V \setminus M]$ contains
$t$ disjoint stars, we have that $2$-\tdmod$(G) \le |M| \le \text{poly}(n)$, as
required.

\medskip
\noindent
\emph{Degeneracy}. Let us prove that $G$ is $4$-degenerate.
Observe that every vertex in $\C$ has degree at most $4$ (three neighbors
in $\X$ and one in $\R$). Thus, every ordering of $V(G)$ of
the form $(\C,\R,\Y,\alpha,\X,\A)$ (with arbitrary order within each set) is a $4$-elimination order of $G$, that is, the vertices of $G$ can be removed according to this ordering so that the current first vertex has at most 4 neighbors in the current graph. \end{proof}

\section{Excluding polynomial kernels parameterizing by $\tw$ or $\td$}
\label{sec:neg2}
Our objective in this section is to show that the meta-result of Gajarsk{\`y} et al.~\cite{gajarsky2013kernelization} cannot be improved by replacing $c$-\tdmod with $\td$, even when restricting ourselves to planar graphs of bounded maximum degree.

%\fixme{quelque part recall that L/tw denote dec pb param py tw of the graph}

When proving lower bounds on the size of kernel for a parameter $\kappa$ such that $\kappa(\bigcup G_i) \le \text{poly}(\max(\kappa(G_i)))$ (such as $\tw$ or $\td$),
compositions are generally simple, as taking the union of the input
graphs preserves the parameter as required (which is obviously not
true, for example, when parameterizing by the size of a solution).
%For example for $IS$,
%given $t$ input $I_i = (G_i,k_i)$ where we have to decide if $\alpha(G_i)
%\ge k_i$, defining $I = (\UG_i,\sum k_i)$ seems reasonable.
However, there is a problem occurring when proving that if the large
union graph is a \textsc{Yes}-instance, then every (or there exists, depending
if we are designing an \AND- or \OR-composition) instance is a \textsc{Yes}-instance, as the sizes of the solutions in the different $G_i$'s are not
necessarily balanced. This explains why in~\cite{bodlaender2009problems} and in
this work we introduce \emph{bounded} versions of decision problems, where we
already know that either a small solution exists, or there
is no larger solution.
%\fixme{laisser (et finir) ce blabla?}

Let us now explain in detail how Bodlaender et al.~\cite{bodlaender2009problems} prove that several problems
(including \IS~and \DS) do not admit a polynomial kernel
parameterized by $\tw$ unless $\text{NP} \subseteq \text{coNP}/\text{poly}$. To that end, they
first define a \emph{refinement} problem, where the input of the classical
problem is augmented with a witness $I$ (corresponding to a subset of
vertices or edges), and the question is to decide whether there exists a
solution of size $|I|+1$ (or $|I|-1$ for a minimization problem).
For example, in the \IS-\REF problem, given a graph $G$ and an independent set $I$,
the question is to decide whether $G$ has an independent set of size
$|I|+1$.
Then, they show that
\begin{enumerate}
\item \IS-\REF is Karp \NP-hard (by a Karp reduction
  from \IS that simply adds $k-1$ independent vertices connected to all
  the old vertices),
\item \IS-\REF/$\tw$ is \OR-compositional,
\item \IS-\REF/$\tw$ does not admit a polynomial kernel unless $\text{NP} \subseteq \text{coNP}/\text{poly}$  (which is a direct consequence of the two previous
points using Theorem~\ref{thm:compo}), and
\item \IS/$\tw$ does not admit a polynomial kernel unless $\text{NP} \subseteq \text{coNP}/\text{poly}$  (by simply observing that \IS-\REF/$\tw$ $\le_{\PPT}$ \IS/$\tw$ and using Theorem~\ref{thm:ppt}).%\footnote{Theorem~\ref{} was stated after \cite{},  but the proofs are the same}
\end{enumerate}

A drawback of this approach is that we lose planarity in Step 1. To obtain the same
results for planar graphs, we propose the following modification of the
previous approach, where we replace the positive witness by an upper
bound (or a lower bound, for minimization problems), and use an \AND-composition instead.

In the following, $\Pi_{\C}$ will denote any NP optimization graph problem where input graphs belong to a graph class $\C$ and,
and $\Pi^{\dec}_{\C}$ its associated decision problem (given $G \in \C$ and $k$, we have to decide whether $\opt(G) \ge k$ for a maximization problem, or $\opt(G) \le k$ for a minimization one).

\begin{definition}[Decision problem with negative witness]
Given a maximization (resp. minimization) problem $\Pi_{\C}$, we define $\Pi_{\C}^{\sup}$ (resp $\Pi_{\C}^{\inf}$) as follows:
\begin{itemize}
\item[] \textbf{Input}: An instance $(G,k)$ of $\Pi^{\dec}_{\C}$ such that $\opt(G) \le k$ (resp. $\opt(G) \ge k$).
\item[] \textbf{Question}: Decide whether $\opt(G)=k$.
\end{itemize}
\end{definition}
%\fixme{
%*ya qu'une bonne façon d'associer pb de dec à opt : pour max c'est opt >= k, sinon c'est pas dans NP (meme si pas ininteressant)
%*de meme on définit pi sup que pour max, sinon le pb est pas dans NP
%*Pi sup se réduit au pb classique opt >= k, donc ouf}
%\fixme{dire que sup c'est que pour maximization sinon c'est pas dans NP ?}
%\fixme{on a nulle part besoin de la monotinoie genre (G,k) oui implique (G,k-1) oui ?}

\begin{definition}
We say that an optimization problem is \emph{additive} if for every two graphs $G_1$ and $G_2$, $\opt(G_1 \cup G_2) = \opt(G_1)+\opt(G_2)$.
\end{definition}

%% \begin{definition}
%% We say that a decision graph problem $\Pi$ is prop1\fixme{ça doit déjà porter un nom} iff for any $a \ge 0$ and $b \ge 0$ and any graphs $G_1$ and $G_2$,
%% \begin{enumerate}
%% \item $(((G_1,a) \in \Pi_{YES}) \mbox{ AND } ((G_2,b) \in \Pi_{YES}) \Rightarrow (G_1 \cup G_2,a+b) \in \Pi_{YES})$
%% \item $(G_1 \cup G_2, a+b+1) \in \Pi_{YES} \Rightarrow ((G_1,a+1) \in \Pi_{YES} \mbox{ OR } (G_2,b+1) \in \Pi_{YES})$
%% \end{enumerate}
%% \end{definition}
Observe that many classical optimization problems are additive, like \IS or \DS.

%% \fixme{remplacer td par param stable par union ? pas sur car pour
%%   versions connectées on a que td, ou sinon ce que doit verif le param
%% est pénible à expliquer}
\begin{proposition}
Let $\C$ be a graph class stable under disjoint union and let $\Pi_\C$ be an additive optimization problem.
Then $\Pi^{\sup}_\C/\td$ and $\Pi^{\inf}_\C/\td$ are \AND-compositional.
%\fixme{on pourrait faire cross-compo mais inutile}
\end{proposition}
\begin{proof}
Let us only prove the result for $\Pi^{\sup}_\C$, as the proof is similar for $\Pi^{\inf}_\C$.
Let $t$ be an integer and let, for $i \le t$, $((G_i,k_i),\td)$ be an instance of $\Pi_\C^{\sup}$, where $\td(G_i)=\td$.
Let $G'$ be the disjoint union of the $G_i$'s and let $k'=\sum_{i=1}^t k_i$. We have $\td(G') \le \max(\td(G_i))$.
As $\C$ is stable under disjoint union, to verify that $(G',k')$ is an instance of
$\Pi_\C^{\sup}$ it only remains to prove that $\opt(G') \le k'$.
%Observe that for any subset of instances $i_{l_1},\dots,i_{l_{t'}}$,
However, as $\Pi$ is additive, we have $\opt(G')=\sum_{1 \le i \le t} \opt(G_i) \le k'$.

It remains to verify that $\opt(G')=k' \Leftrightarrow \forall i, \opt(G_i)=k_i$.
$\Leftarrow:$ is straightforward by the additivity of $\Pi$.
$\Rightarrow:$ Let us suppose that $\sum_{1 \le i \le t} \opt(G_i)=k'$, and let $\ell \le t$. Again, as for every $i$ we have $\opt(G_i) \le k_i$, we deduce $\opt(G_\ell) \ge k_\ell$, and thus $\opt(G_\ell)=k_\ell$.
\end{proof}

%% \begin{proposition}
%% $\IS_{\inf}$ on planar graphs is \AND-compositional parameterized by $tw$
%% \end{proposition}
%% \begin{proof}
%% The proof is straightforward as given $t$ input graphs $G_i$ it is sufficient to define $G'$ as the union of
%% these graphs and $k'=kt$. This is an instance of
%% $\IS_{binf}$ as $G'$ remains planar and $\alpha(G') \le kt$. As
%% $\alpha(G') = kt$ iff $\forall i \le t, \alpha(G_i)=k$, and $\tw(G')
%% \le \max \tw(G_i)$, we get the desired result.
%% \end{proof}

According to Theorem~\ref{thm:compo}, %as \IS_inf est karp,
                                %NPhard, et AND compositional, verif
                                %que pour and compositional ya bien
                                %aussi besoin de karp
we get the following results. Note that as $\Pi_\C^{\sup} \le_{\PPT} \Pi^{\dec}_\C$, we use Theorem~\ref{thm:ppt} and also state the result for $\Pi^{\dec}_\C$ in the next theorem.
\begin{theorem}
Let $\C$ be a graph class stable under disjoint union.
\begin{itemize}
\item[$\bullet$] For any additive maximization problem
$\Pi$ such that \emph{$\Pi_\C^{\sup}$} is Karp \NP-hard,
 \emph{ $\Pi^{\sup}_\C/\td$} (and thus \emph{$\Pi^{\dec}_\C/\td$}) does not admit a polynomial
  kernel unless \emph{$\text{NP} \subseteq \text{coNP}/\text{poly}$}.
\item[$\bullet$] For any additive minimization problem
$\Pi$ such that \emph{$\Pi_\C^{\inf}$} is Karp \NP-hard,
  \emph{$\Pi^{\inf}_\C/\td$} (and thus \emph{$\Pi^{\dec}_\C/\td$}) does not admit a polynomial
  kernel unless \emph{$\text{NP} \subseteq \text{coNP}/\text{poly}$}.
\end{itemize}
\end{theorem}

%% \fixme{soit on utilise les reducs folkore pour dire \IS <=ppt PI2 et
%%   donc no kernel pour PI2 dans planaire,
%% soit on reformule previous truc avec un PIinf generique (et lemme si
%% PIinf karp NP hard dans C alors no poly kernel pour Pi dans C pour p,
%% ou C et p n'importe quelle classe et p qui resistent à l'union) Peut
%% avoir l'air mieux, sauf qu'apres faut verif que chaque folkore reduc
%% preserve le cote binf, alors que dasn l'autre cas on regarde que param
%% preserv reduc
%% soit mieux : on formule le max de truc avec PIinf generque (et version connectée), après on applique avec \ISinf planaire, on en déduit boite noire \IS, et après reduc folkore entre \IS et les autres}.

%% \fixme{si on fait res generique pour no poly kernel pour PI binf, ya donc 2 types de Pi binf, et soit on les ecrit en max S tq G[S] verifi Pi hereditaire et min S tq G[V moins S] verifi Pi hereitaire,
%% soit on l'écrit comme dans on pb without poly avec "un langage tq G1 U .. U Gt, kt+1  => existe Gi, k+1  et G1,k,  G2,k implique G1 U G2, 2k" }

According to the previous theorem, it turns out that to exclude a kernel by
treedepth we only have to prove that the ``negative witness'' version of
a decision problem is Karp \NP-hard, which is usually almost already given by
classical reductions.

\begin{proposition}
\textsc{IS}$^{\sup}_{\C}$ is Karp \NP-hard, where $\C$ is the class of planar graphs of maximum degree at most $4$.
\end{proposition}
\begin{proof}
It is sufficient to observe that the reduction from planar $3$-\textsc{Sat}(5) (where each variable appears in at most $5$ clauses) to planar \IS provided in \cite{lichtenstein1982planar}
is in fact a reduction to planar \textsc{IS}$^{\sup}$. For the sake of completeness let us recall this reduction.

%This is a folklore reduction from SAT to \IS (see~\cite{papdim book}
%  ex 9.5.16 for example).
%
An instance of planar $3$-\textsc{Sat}(5) is described by a set $C$ of $m$
clauses $c_i$ and a set $X$ of $n$ variables $x_j$, where each $c_i$ contains
exactly three literals (where a literal is of the form $x_j$ or $\bar{x}_j$). We can clearly assume that each variable appears both positively and negatively.
Consider the incidence graph (which is planar) $G=(V,E)$ that has $V=C \cup X$ and has an
edge from a variable to a clause if the variable or its negation
appear in the clause. We define $G'$ by replacing each $c_i$ with a
triangle $T_{c_i}$ (each vertex of the triangle is associated with a literal of
the clause) and each $x_j$ with a cycle on 4 vertices
$C_{x_j}=\{v^1_{x_j},v^1_{\bar{x}_j},v^2_{x_j},v^2_{\bar{x}_j}\}$
with edges $\{v^t_{x_j},v^{t'}_{\bar{x}_j} \mid t,t' \in \{1,2\}\}$.
Then, we add edges between $T_{c_i}$ and $C_{x_j}$ in the following
way. If $T_{c_i}$ contains a vertex $v$ corresponding to $x_\ell$
(resp. $\bar{x}_\ell$), we add exactly one of the two edges $\{v,v^t_{\bar{x}_\ell}\}$
(resp. $\{v,v^t_{x_\ell}\}$) where $t \in \{1,2\}$, by choosing  $t$ such that $G'$ remains planar and of maximum degree 4. Indeed, as each
variable appears in at most $5$ clauses and appears both positively and negatively, it is always possible to
embed the $C_{x_j}$'s such that the at most $5$ edges of the form $\{v_1,v_2\}$
with $v_1$ in a clause triangle and $v_2 \in C_{x_j}$ do not cross when connecting to $C_{x_j}$, while avoiding to connect more than 2 edges to the same vertex in $C_{x_j}$.
%It is easy to verify that the edges can be added in such a way that the maximum degree of the constructed planar graph is at most $4$.

Observe that $G'$ can be partitioned into $m$ triangles and $n$
$C_4$'s, implying $\alpha(G') \le m+2n$. As $\alpha(G')=m+2n$ if and only if the
original instance of planar $3$-\textsc{Sat}(5) is satisfiable, we get the
desired result.
\end{proof}

%% \begin{proposition}
%% $\CVC^{\inf}_{\C}$ is Karp \NP-hard, where $\C$ is the class of planar graphs of maximum degree $4$.
%% \end{proposition}
%% \begin{proof}
%% We will reuse the above reduction of \cite{lichtenstein1982planar} from planar 3-SAT(5) to planar \VC, using the same
%% notations.
%% We emphasize that according to
%% \cite{lichtenstein1982planar} the planar graph $G$ of maximum degree $4$ obtained in the above
%% reduction has the additional property that we can draw a Hamiltonian
%% path $P$ covering all variable vertices (\ie considering $C_{x_j}$ as a
%% vertex) and still get a planar graph. Without loss of generality we suppose that
%% $P=(C_{x_1},\dots,C_{x_n})$.

%% Let us now describe how we modify $G$.
%% For each cycle $C_{x_j}$ we add vertices $r_j,r'_j$ where $r_j$ is connected to the $4$
%% vertices of $C_{x_j}$, and $r'_j$ is only connected to $r_j$ (so that any \CVC
%% has to choose $r_j$). Now, for $i \in [n-1]$, let us consider edge
%% $\{C_{x_i},C_{x_{i+1}}\}$ of $P$, and let \fixme{given Ci, 2 cas :
%% si le P arrive et repart entre les 2 meme sommets du chemin : un seul
%% gadget a (+ pending) connecté à ces 2 sommets, sinon 2}
%% %$a,b$ we add two vertices $s_i,s'_j$ where $

%% %This is a folklore reduction from SAT to \IS (see~\cite{papdim book}
%% %  ex 9.5.16 for example).
%% %
%% \end{proof}

%\fixme{autre preuve avec some simplified NP compleness proof for graph
%  pb ou ils montrent que \VC planaire dur avec gadget et instance du
%  type binf}
Note that it could be tempting to directly conclude that
\textsc{IS}$^{\sup}$ is \NP-hard on planar graphs using the classical Turing reduction from
\IS on planar graphs (\ie starting with an instance of \textsc{IS}$^{\sup}$ with
$k=n$, and decreasing $k$ while the oracle for \textsc{IS}$^{\sup}$ answers `\textsc{No}').
However, this would only prove that \textsc{IS}$^{\sup}$ is Turing \NP-hard on planar
graphs, which is not sufficient to use Theorem~\ref{thm:compo} that requires Karp
\NP-hardness, even if this detail is generally not mentioned in the statement.

As \IS is an additive problem, we immediately deduce the following corollary.
\begin{corollary}\label{corollary:ISCVC}
\textsc{IS}$/$\td\ does not admit a polynomial kernel on planar graphs of maximum degree at most $4$ unless \emph{$\text{NP} \subseteq \text{coNP}/\text{poly}$}.%, where $\C$ is the class of
\end{corollary}

To propagate the previous result to almost all
problems covered by the meta-result of Gajarsk{\`y} et al.~\cite{gajarsky2013kernelization}, we will use Theorem~\ref{thm:ppt} on folklore reductions and verify that the treedepth is polynomially preserved.
To avoid tedious enumeration of problems, we restrict our attention to problems mentioned in Corollary~\ref{corollary:nokernel} below.
Note that for problems like \textsc{Longest path} and \textsc{Treewidth} where $\opt(G_1 \cup G_2) = \max(\opt(G_1),\opt(G_2))$, we also get that a polynomial kernel is unlikely to exist
on planar graphs, as taking the union of input graphs provides a trivial \OR-composition.
%et que pour les pbs ou opt(G1 U G2) = max(opt(G1),opt(G2)) ou min, or compo trivialle par union donne result.
%  (on oublie HC : and compo à la main avec arete qu'on doit prendre dans chaque graph)}
%si on a des results interessants pour d'autres pb on les rajoutera !

\begin{corollary}\label{corollary:nokernel}
The following problems do not admit a polynomial kernel when parameterized by $\td$ or $\tw$ on planar graphs of bounded maximum degree unless \emph{$\text{NP} \subseteq \text{coNP}/\text{poly}$}:
\VC, \FVS, \OCT, \DS, $r$-\DS, \textsc{Chordal Vertex Deletion}, \textsc{Induced Matching}.
%% \DS -> Efficient Dominating Set : remplacer chaque arete u-v par
%%   u-a-b-v + x-a + x-b  : => >= 1 sommet par arete verfi que k <=> m+k
%%   (sens retour, soit X = vrais sommets pris, Y = voisins, Z = reste, si
%%   |Z|=1 : on a dom de taille m+|X|+1, on en déduit bien dominant taille
%%   X+1 en prendant X U Z)
%% Edge Dominating Set : avec igna
%%   Induced d-Degree Subgraph, Min Leaf Spanning Tree, Max Full Degree Spanning Tree à voir
\end{corollary}

\begin{proof} We split the proof into several problems.

\vspace{.2cm}
\noindent
\textsc{FVS, OCT, \DS.}
For these three problems we use the same folkore reduction. Given an input $(G,k)$ of \VC, we define $G'$ by adding for each edge $\{u,v\}$ of $G$ a vertex $x_{uv}$,
and two edges $\{x_{uv},u\}$ and $\{x_{uv},v\}$. It is straightforward
to verify that $G$ admits a vertex cover of size at most $k$ if and only if $G'$ admits a
FVS (or OCT, or \DS) of size at most $k$. As $\td(G') \le \td(G)+1$
(in the treedepth decomposition of $G$, for each vertex $u$ in $G$ we add
degree of $u$ new leaves attached to $u$), this is a PPT reduction.

\vspace{.2cm}
\noindent
\textsc{Chordal vertex deletion.} We use almost the same reduction as above: given an input $(G,k)$
of \VC, we define $G'$ by adding for each edge $\{u,v\}$ of $G$ two
vertices $x^1_{uv}$, $x^2_{uv}$, and edges
$\{u,x^1_{uv}\},\{x^1_{uv},x^2_{uv}\}, \{x^2_{uv},v\}$.

\vspace{.2cm}
\noindent
\textsc{$r$-\DS.}
Given an input $(G,k)$ of \DS, we define $G'$ by adding a pendant
$P_{r}$ attached to each vertex of $G$, and we set $k'=k$. As $r$
is constant, the treedepth is polynomially preserved.

\vspace{.2cm}
\noindent
\textsc{Induced Matching.}
Given an input $(G,k)$ of \IS, we define $G'$ by adding a pendant vertex
to each vertex of $G$, and we set $k'=k$. \end{proof}

Note that Corollary~\ref{corollary:nokernel} does not apply for $\C^1$,
defined as the class of {\sl connected} planar graphs of bounded maximum degree, as $\C^1$ is obviously not stable under disjoint union.
However, there are two ways to get the same result for all problems of
Corollary~\ref{corollary:nokernel} for $\C^1$, using the following observations.

If a problem $\Pi$ is solvable on planar graphs in time $\O^*(c^{\td})$ for some constant $c$ (which is often the case~\cite{DPBF10,CNP11,RueST14}),
we can show using an ad-hoc argument (\ie depending on the problem) that $\Pi_{\C}/$\td $\le_{\PPT} \Pi_{\C^1}/$\td. Let us illustrate this for \IS.
Given an instance $(G,k)$ of \textsc{IS}$_{\C}/\td$ having $k_1 \le k$ connected components $X_i$, let $v_i \in X_i$ be a vertex on the outer face of $X_i$.
For every $i$ we add two vertices $a_i,b_i$ and edges $\{a_i,b_i\}$, $\{b_i,v_i\}$ so that any optimal solution takes $a_i$ and not $b_i$, and we connect all the $b_i$'s by creating a path.
We get a graph $G'$, and we set $k'=k+k_1$ so that $(G,k)$ and
$(G',k')$ are equivalent. $G'$ is still planar of bounded maximum degree and
$\td(G') \le \td(G)\lceil \log_2(k_1+1) \rceil )$ (as $\td(P_{k_1})
\le \lceil \log_2(k_1+1) \rceil$).
If $\log_2(k_1) \le \td(G)$ then this is PPT reduction. Otherwise, we
can solve the original input in polynomial time and also get the reduction.

Alternatively, it is generally possible to directly cross-compose from $\Pi^{\inf}_{\C^1}$ (or $\Pi^{\sup}_{\C^1}$) to $\Pi_{\C^1}/\td$ using again an ad-hoc argument to connect the graph.
Given the $t$ instances $\{G_i\}$ of $\Pi^{\inf}_{\C^1}$, we define $G'$ by adding a dummy vertex $v_i$ on the outer face of $G_i$ and connecting the $v_i$'s by creating a path.
If the dummy vertices are added such that $G'$ is a \textsc{yes}-instance if and only if all the $G_i$'s are \textsc{yes}-instances,  we have an \AND-cross-composition, as again $\td(G') \le \max(\td(G_i))\lceil \log_2(k_1+1) \rceil )$, and this $\log$ factor is allowed in the parameter of a cross-composition.

\section{Concluding remarks and further research}
\label{sec:conclusions}

In this article we studied the existence of polynomial kernels for problems parameterized by the size of a $c$-treedepth modulator, on graphs that are not necessarily sparse. On the positive side, we proved that \textsc{Vertex Cover} (or equivalently, \textsc{Independent Set})  parameterized by the size $x$ of a $c$-treedepth  modulator admits a polynomial kernel on general graphs with $x^{2^{\O(c^2)}}$ vertices, for every $c \geq 1$. A natural direction is to improve the size of this kernel. Since
\textsc{Vertex Cover} parameterized by the distance to a disjoint collection of cliques of size at most $c$ does not admit a kernel with $\O(x^{c-\epsilon})$ vertices unless $\text{NP} \subseteq \text{coNP}/\text{poly}$~\cite{majumdar2015kernels}, and since a clique of size $c$ has treedepth $c$, the same lower bound applies to our parameterization; in particular, this rules out the existence of a {\sl uniform} kernel. However, there is still a large gap between both bounds, hence there should be some room for improvement.

On the negative side, we proved that \textsc{Dominating Set}  parameterized by the size of a $c$-treedepth modulator does not admit a polynomial kernel on $4$-degenerate graphs for any $c \geq 2$. As \textsc{Dominating Set} with this parameterization admits a polynomial kernel on nowhere dense graphs~\cite{gajarsky2013kernelization}, it follows that sparse graphs constitute the border for the existence of polynomial kernels for \textsc{Dominating Set}. This leads us to the following natural question: are there smaller parameters for which \textsc{Dominating Set} still admits polynomial kernels on sparse graphs? Since considering as parameter the treedepth of the input graph does not allow for polynomial kernels (see Section~\ref{sec:neg2}), we may consider as parameter the size $x$ of a vertex set whose removal results in a graph of treedepth at most $b(x)$, for a function $b$ that is not necessarily constant. We prove in Appendix~\ref{sec:negative-DS-pushed} that \textsc{Dominating Set} does {\sl not} admit polynomial kernels on graphs of bounded expansion for $b(x) = \Omega (\log x)$, unless $\text{NP} \subseteq \text{coNP}/\text{poly}$. On the other hand, by combining the approach of Garnero et al.~\cite{garnero2015explicit} to obtain explicit kernels via dynamic programming  with the techniques of Gajarsk{\`y} et al.~\cite{gajarsky2013kernelization} on graphs of bounded expansion, it can be shown -- we omit the details here -- that \textsc{Dominating Set} admits a polynomial kernel  for $b(x) = \O(\log \log \log x)$ on graphs of bounded expansion whose expansion function $f$ is not too ``large''\footnote{That is, the function $F$ that bounds the grad with rank $d$ of the graphs in the family, see~\cite{sparsity}.}, namely $f(d) = 2^{\O(d)}$. While this result is somehow anecdotal, we think that it may be the starting point for a systematic study of this topic.%, which can be seen as a generalization of the results of

We leave as an open question the existence of polynomial kernels on general graphs for other natural problems  parameterized by the size of a treedepth modulator, such as \textsc{Feedback Vertex Set} (already studied by Jansen et al.~\cite{jansen2014parameter}  under several parameterizations) or \textsc{Odd Cycle Transversal}.  In the spirit of the recent meta-kernelization results on sparse graphs~\cite{FLST10,KLPRRSS16,gajarsky2013kernelization,BFL+09,garnero2015explicit}, it would be interesting to find generic conditions for a problem to admit polynomial kernels on general graphs with this parameter. To the best of our knowledge, the only meta-kernelization results with structural parameters on general graphs are the work of Ganian et al.~\cite{ganian2016meta}, where
the parameter is the minimum number of parts $V_i$'s required in a vertex partition such that every $V_i$ is a module (for every $v \in V \setminus V_i$, either all or no vertex of $V_i$ is adjacent to $v$) and $G[V_i]$ has bounded rankwidth\footnote{Note that our kernel for \VC/$c$-\tdmod is not covered by the meta-result of~\cite{ganian2016meta}. Indeed, given a $c$-treedepth modulator $X=\{v_i \mid i \in [|X|]\}$, we could define a partition of $V(G)$ with $V_i = \{v_i\}$ for $i \in [|X|]$ and $V_{|X|+1} = V(G)\setminus X$.
The number of parts is polynomial in $|X|$,  each satisfying the rankwidth condition: $\rw(V_{|X|+1}) \le \tw(V_{|X|+1})+1 \le \td(V_{|X|+1})+1 \le c+1$. However, $V_{|X|+1}$ is not a module in general.}, and the further extension provided by Eiben et al.~\cite{EibenGS15}, which also subsumes the meta-kernelization framework of Gajarsk{\`y} et al.~\cite{gajarsky2013kernelization}.

Finally, it is worth studying whether the $2^c$-approximation algorithm of~\cite{gajarsky2013kernelization} for computing a $c$-treedepth modulator can be improved (maybe, using the algorithm of Reidl et al.~\cite{ReidlRVS14} for computing treedepth), or whether a  lower bound can be proved.

\medskip

{\small \noindent \textbf{Acknowledgement}. We would like to thank the anonymous reviewers  for helpful comments that improved the presentation of the manuscript.}

%\ig{cite article of Jansen at last IPEC?}

%\ig{je vais ecrire un ``skecth of proof'' du resultat suivant:}
%
%
%\begin{proposition}\label{prop:ds/tdmod-positive}
%\DS$_{\C_{BE}}/\log \log \log $-\tdmod admits a polynomial kernel, where $\C_{BE}$ is the class of
%graphs of bounded expansion.
%\end{proposition}
%\begin{proof}
% The idea is to combine the approach of Garnero \emph{et al}.~\cite{garnero2015explicit} to obtain explicit kernels via dynamic programming  with the techniques of Gajarsk{\`y} \emph{et al}.~\cite{gajarsky2013kernelization} on graphs of bounded expansion. The kernelization algorithms works as follows. Given an instance $(G,X,k)$ of \textsc{Dominating Set} such that $\td(G - X) = \O (\log \log \log |X|)$ (as discussed before, we can safely assume that we are just given $|X|$, by using the constant-factor approximation of the modulator given in~\cite{gajarsky2013kernelization}), let for convenience $t = \td(G - X)$. We use the ``bag marking'' algorithm of~\cite{gajarsky2013kernelization}, which is strongly inspired by the algorithm of Kim \emph{et al}.~\cite{KLPRRSS16} for graphs excluding (topological) minors, to obtain in linear time a protrusion decomposition of $G$, where the number of protrusions and their treedepth (and hence, their treewidth as well) is
%\end{proof}
%

%%
%% Bibliography
%%

%% Either use bibtex (recommended),

%\newpage

%{\small
\bibliographystyle{abbrv}
\bibliography{biblio,biblio-Ignasi}
%}

%% .. or use the thebibliography environment explicitely

\newpage

\begin{appendix}

\section{List of problems considered in this article}
\label{ap:problems}

~

\vspace{-.2cm}
\begin{boxedminipage}{.99\textwidth}
\textsc{Independent Set (\IS)}\vspace{.1cm}

\begin{tabular}{ r l }
\textbf{~~~~Instance:} & $(G,k)$ with $G=(V,E)$ a graph and $k$ an integer. \\
\textbf{Question:} & Decide whether $\alpha(G) \ge k$, \\
& \ie if $\exists S \subseteq V$ such that $\forall e \in E, e \nsubseteq S$ and $|S| \ge k$. \\
\end{tabular}
\end{boxedminipage}

\vspace{.4cm}
\begin{boxedminipage}{.99\textwidth}
\textsc{$c$-treedepth modulator Independent Set ($c$-\tdmod-\IS)}\vspace{.1cm}

\begin{tabular}{ r l }
\textbf{~~~~Instance:} & $(G,X,k)$ with $G=(V,E)$ a graph, $X$ a $c$-treedepth modulator, $k \in \mathds{N}$. \\
\textbf{Question:} & Decide whether $\alpha(G)\ge k$.\\
\end{tabular}
\end{boxedminipage}
\vspace{.4cm}

%\vspace{.4cm}
\begin{boxedminipage}{.99\textwidth}
\textsc{Annotated $c$-treedepth modulator Independent Set} (a-$c$-\tdmod-\IS)\vspace{.1cm}

\begin{tabular}{ r l }
\textbf{~~~~Instance:} & $(G,X,k)$ where \\
 & ~\textbullet\ $G=(V,E,\H)$ is a hypergraph structured as follows: $V = X \uplus R$,\\
 & ~~ $E = E_{X,R} \uplus E_{R,R}$ is a set of edges where edges in $E_{A,B}$ have one \\
 & ~~  endpoint in $A$ and the other in $B$, and $\H \le 2^{X}$ is a set of \\
 & ~~  hyperedges where each $H \in \H$ is entirely contained in $X$.\\
 & ~\textbullet\ $X$ is a $c$-treedepth modulator (as $G[V\setminus X]$ is not a hypergraph, \\
 & ~~ its treedepth is correctly defined and we have $\td(V\setminus X) \le c$).\\
 & ~\textbullet\ $k$ is a positive integer.\vspace{.1cm}\\

\textbf{Question:} & Decide whether $\alpha(G) \ge k$ (an independent set in a hypergraph is a \\
 & subset of vertices that does not contain any hyperedge, corresponding  \\
 & here to a subset $S \subseteq V$ such that for every $h \in E \cup H$, $h \nsubseteq S$).\\
\end{tabular}
\end{boxedminipage}
\vspace{.4cm}

%\vspace{.4cm}
\begin{boxedminipage}{.99\textwidth}
\textsc{Vertex Cover (\VC)}\vspace{.1cm}

\begin{tabular}{ r l }
\textbf{~~~~Instance:} & $(G,k)$ with $G=(V,E)$ a graph and $k$ an integer \\
\textbf{Question:} & decide whether $G$ has a vertex cover of size at most $k$, \ie if \\
 &  there exists $S \subseteq V$ such that $\forall e \in E, e \cap S \neq \emptyset$ and $|S| \le k$.
\end{tabular}
\end{boxedminipage}
\vspace{.4cm}

%
%
%\vspace{.4cm}
%\begin{boxedminipage}{.99\textwidth}
%\textsc{Connected Vertex Cover (\CVC)}\vspace{.1cm}
%
%\begin{tabular}{ r l }
%\textbf{~~~~Instance:} & A graph $G$ and a positive integer $k$.\\
%\textbf{Question:} & Decide whether $G$ has a connected vertex cover of size at most $k$ (if $G$ has \\
% & several connected component, then we obviously require one connected vertex \\
% & cover for each connected component)\\
%\end{tabular}
%\end{boxedminipage}
%\vspace{.4cm}
%

%\vspace{.4cm}
\begin{boxedminipage}{.99\textwidth}
\textsc{Feedback Vertex Set (FVS)}\vspace{.1cm}

\begin{tabular}{ r l }
\textbf{~~~~Instance:} & $(G,k)$ with $G=(V,E)$ a graph and $k$ an integer.\\
\textbf{Question:} & Decide whether $G$ has a feedback vertex set of size at  most $k$, \ie if\\
&  there exists $S \subseteq V$ such that $G[V\setminus S]$ is a forest and $|S| \le k$ \\
\end{tabular}
\end{boxedminipage}
\vspace{.4cm}

%\vspace{.4cm}
\begin{boxedminipage}{.99\textwidth}
\textsc{Dominating Set (\DS)}\vspace{.1cm}

\begin{tabular}{ r l }
\textbf{~~~~Instance:} & $(G,k)$ with $G=(V,E)$ a graph and $k$ an integer.\\
\textbf{Question:} & Decide whether $G$ has a dominating set of size at most $k$,  \ie if there\\
& exists $S \subseteq V$ such that $\forall u \in V\setminus S, \exists v \in S \mid \{u,v\} \in E$ and $|S| \le k$.
\end{tabular}
\end{boxedminipage}
\vspace{.4cm}

%\vspace{.4cm}
\begin{boxedminipage}{.99\textwidth}
\textsc{Red Blue Dominating Set (\RBDS)}\vspace{.1cm}

\begin{tabular}{ r l }
\textbf{~~~~Instance:} & $(U,W,E,k)$ where $(U,W,E)$ is a bipartite graph and $k$ is an integer.\\
\textbf{Question:} & Decide whether there exists $S \subseteq W$ such that $N(S)=U$ and $|S| \le k$. \\
\end{tabular}
\end{boxedminipage}
\vspace{.4cm}

%% \begin{definition}[coloured \RBDS (col-\RBDS)]
%% \begin{itemize}
%% \item Input: an input  $(U,W,E,k,col)$ (\fixme{avec Ui Wi sommets, E
%%   aretes col(w) dans 1 k}
%% \item Question:
%% \end{itemize}
%% \end{definition}

%% \begin{definition}[coloured \RBDS with $\Deta_W \le c$ (col-\RBDS-$\Delta_W \le c$)]
%% \begin{itemize}
%% \item Input: an input  $(U,W,E,k,col)$ (\fixme{avec Ui Wi sommets, E  aretes col(w) dans 1 k} such that $\forall v \in W, d(v) \le c$
%% \item Question:
%% \end{itemize}
%% \end{definition}

\section{Stronger negative results for \textsc{Dominating Set}}
\label{sec:negative-DS-pushed}

In this section we rule out the existence of polynomial kernels for \textsc{Dominating Set} on graphs of bounded expansion for a parameter that is smaller than the size of a $c$-treedepth modulator. Let $b: \N \rightarrow \R$ be a function. A \emph{$b$-treedepth modulator} of a graph $G=(V,E)$ is a subset of vertices $X \subseteq V$ such that $\td(G[V\setminus X]) \le b(|X|)$, and we denote by
$b$-$\tdmod(G)$ the size of a smallest $b$-treedepth modulator of $G$. Note that, in the particular case where the function $b$ is constantly equal to a positive integer $c$, $b$-treedepth modulators correspond exactly to $c$-treedepth modulators. In the following proposition we show that $\DS$ does not admit polynomial kernels on graphs of bounded expansion parameterized by the size of a $b$-treedepth modulator with $b(x) = \Omega (\log x)$.

\begin{proposition}\label{prop:ds/tdmod-pushed}
$\DS/\log$-$\tdmod$ does not admit a polynomial kernel on graphs of bounded expansion unless
\emph{$\text{NP} \subseteq \text{coNP}/\text{poly}$}.
\end{proposition}
%\ig{we have to argue that the treedepth of $G - X$ is at most $\log(|X|)$ (instead of $\log n$ as before). Fortunately, it is true, and the same proof works}
%\vspace{-.2cm}
\begin{proof} As in the proof of Proposition~\ref{prop:ds/tdmod},
it is sufficient to prove that $\RBDS/U \le_{\PPT}
\DS_{\C_{BE}}/\log$-\tdmod, as $\RBDS/U$ (and even $\RBDS/(k+U)$) does not
admit a polynomial kernel unless $\text{NP} \subseteq \text{coNP}/\text{poly}$~\cite{dom2014kernelization}.
Let $(G=(U,W,E),k)$ be an instance of $\RBDS/U$ with $u=|U|$, $w=|W|$, and $m=|E|$. Let
$G'=G^{3u\text{-sub}} \cup \tilde{G}$, where $\tilde{G}$ is a square grid on
$(3u+1)^4$ vertices plus a vertex $\alpha$ connected to all vertices of
this grid. $G$ has a RB-dominating set of size $k$ if and only if $G'$ has a
dominating set of size $k+um+1$ (as we also have to take $\alpha$ in
the solutions of $G'$). Let $X=U \cup \tilde{V}$ where $\tilde{V}$ is
the vertex set of $\tilde{G}$. Note that $G'[V'\setminus X]$ is a disjoint
collection of spiders $S_v$ (one rooted at each $v \in W$) of height
$3u+1$. As $\td(S_v) \le 1+\td(P_{3u}) \le 2+ \log(3u+1) \le 2\log(3u+1) = \O(\log |X|)$ and $\td(G') \ge
\td(\tilde{G}) \ge (3u+1)^2$, we get that $X$ is a log-treedepth modulator of
$G'$, and that $\log$-\tdmod$(G') \le |X| \le \text{poly}(|U|)$. To summarize, we added a large grid to artificially increase
the treedepth of $G'$. Moreover, observe that we could not reduce
directly from \DS/\VC as before, as we need a lower bound depending on
\VC of the form $\frac{1}{\text{poly}(u)}$.
Let us finally verify that $G'$ has bounded expansion. As $\tilde{G}$
is an apex graph, it has bounded expansion (as, for instance, planar graphs are well-known to have bounded expansion, and the addition of an apex vertex preserves this property),
and thus it remains to verify that
$G^{3u\text{-sub}}$ has bounded expansion. Let $K=K^{3u\text{-sub}}_{u,w}$. As $G^{3u\text{-sub}} \subseteq K$ as a subgraph, it is sufficient to prove that $K$ verifies the condition of bounded expansion.

To that end, we will prove that $\tilde\nabla_r(K)\leq r+2$, where $\tilde\nabla_r(G)$ denotes the density of a depth-$r$ topological minor using the notation of~\cite{sparsity}. Let $H$ be a depth-$r$ topological minor of $K$.
If $r<3u$, then $H$ is clearly $2$-degenerate. If $r\geq 3u$, observe that every vertex of $K$ that was originally in $W$  has degree at most $u$, and every subdivision vertex (\ie a vertex which is not already a vertex of $K_{u,w}$) has degree $2$. As in a topological minor a vertex cannot have a higher degree than in the original graph, and $K$ is bipartite, we conclude that $H$ is $u$-degenerate. Hence, taking into account both cases, we have that $\tilde\nabla_r(K)\leq r+2$. This proves that the class $\{K_{u,w}^{3u\text{-sub}} : u,w\in \mathbb{N}\}$ has bounded expansion. Thus, this is a PPT reduction and we get the desired result. \end{proof}

\end{appendix}

\end{document}